\documentclass[11pt,a4paper]{article}

\usepackage[T1]{fontenc}
\usepackage[utf8]{inputenc}
\usepackage{lmodern}
\usepackage[margin=25mm]{geometry}
\usepackage{amsmath,amssymb,amsthm}
\usepackage{ifpdf}
\usepackage[hidelinks]{hyperref}
\usepackage[capitalize,nameinlink]{cleveref}
\usepackage{microtype}

\numberwithin{equation}{section}
\newtheorem{theorem}{Theorem}[section]
\newtheorem{lemma}[theorem]{Lemma}
\newtheorem{corollary}[theorem]{Corollary}
\newtheorem{proposition}[theorem]{Proposition}
\newtheorem{definition}[theorem]{Definition}

\newenvironment{keywords}
  {\par\medskip\noindent\textbf{Key words.}\ \ignorespaces}
  {\par\medskip}
\newenvironment{MSCcodes}
  {\par\noindent\textbf{MSC codes.}\ \ignorespaces}
  {\par\medskip}
\newcommand{\email}[1]{\protect\href{mailto:#1}{#1}}
\newcommand{\funding}[1]{\textbf{Funding:}\ #1}
\date{}
\usepackage{lipsum}
\usepackage{amsfonts}
\usepackage{graphicx}
\usepackage{epstopdf}
\usepackage{algorithmic}
\usepackage[abbreviations]{glossaries-extra}
\ifpdf
  \DeclareGraphicsExtensions{.eps,.pdf,.png,.jpg}
\else
  \DeclareGraphicsExtensions{.eps}
\fi

\usepackage{enumitem}
\setlist[enumerate]{leftmargin=.5in}
\setlist[itemize]{leftmargin=.5in}

\newcommand{\Rinf}{\mathbb{R} \cup \{+\infty\}}

\DeclareMathOperator{\intdom}{int\,dom}
\DeclareMathOperator*{\argmin}{arg\,min}
\DeclareMathOperator{\rightprox}{\overrightarrow{\operatorname{prox}}}
\DeclareMathOperator{\leftprox}{\overleftarrow{\operatorname{prox}}}

\newabbreviation{nll}{NLL}{negative log-likelihood}
\newabbreviation{em}{EM}{expectation--maximization}
\newabbreviation{mm}{MM}{majorization--minimization}
\newabbreviation{bpg}{BPG}{Bregman proximal gradient}
\newabbreviation{mlem}{MLEM}{maximum-likelihood expectation--maximization}
\newabbreviation{kkt}{KKT}{Karush--Kuhn--Tucker}
\newabbreviation{kl}{KL}{Kullback--Leibler divergence}
\newabbreviation{klp}{K\L{}}{Kurdyka--\L{}ojasiewicz}
\newabbreviation{tv}{TV}{total variation}
\newabbreviation{pet}{PET}{positron emission tomography}
\newabbreviation{ct}{CT}{computed tomography}
\newabbreviation{psf}{PSF}{point spread function}
\newabbreviation{psnr}{PSNR}{peak signal-to-noise ratio}
\newabbreviation{osl}{OSL}{one-step-late}
\newabbreviation{mu}{MU}{multiplicative updates}
\newabbreviation{mapem}{MAP-EM}{maximum a posteriori expectation--maximization}
\newabbreviation{lsc}{lsc}{lower semi-continuous}
\newabbreviation{pdhg}{PDHG}{primal--dual hybrid gradient}
\newabbreviation{snr}{SNR}{signal-to-noise ratio}
\newabbreviation{nrmse}{NRMSE}{normalized root mean square error}
\newabbreviation{rc}{RC}{recovery coefficient}
\newabbreviation{rbpg}{R-BPG}{right-Bregman proximal gradient}

\glsdisablehyper

\makeatletter
\DeclareRobustCommand{\thmcite}[2]{%
  \lbrack
  \begingroup\def\@cite##1##2{##1}\cite{#1}\endgroup
  ,\ #2%
  \rbrack
}
\makeatother

\theoremstyle{remark}

\newtheorem{assumption}[theorem]{Assumption}
\crefname{hypothesis}{Hypothesis}{Hypotheses}
\theoremstyle{plain}

\theoremstyle{remark}

\crefname{fact}{Fact}{Facts}

\title{Right Bregman proximal gradient with application to Poisson inverse problems\thanks{\funding{This work was partially funded by the PRIMES LABEX (ANR-11-LABX-
0063) at the University of Lyon, within the "Investissements d'Avenir”
program (ANR-11-IDEX-0007).}}}

\author{Thibaut Modrzyk\thanks{INSA-Lyon, CREATIS UMR 5220 U1294, Villeurbanne, F-69621, France
  (\email{thibaut.modrzyk@creatis.insa-lyon.fr}, \email{voichita.maxim@creatis.insa-lyon.fr}).}
\and Elie Bretin\thanks{INSA-Lyon, CNRS UMR 5208, Institut Camille Jordan, Villeurbanne, F-69621, France
  (\email{elie.bretin@insa-lyon.fr}).}
\and Voichita Maxim\footnotemark[2]}

\usepackage{amsopn}

\hypersetup{
  pdftitle={Right Bregman proximal gradient with application to Poisson inverse problems},
  pdfauthor={Thibaut Modrzyk, Elie Bretin, Voichita Maxim}
}

\begin{document}

\maketitle

\begin{abstract}
  We introduce a novel Bregman proximal algorithm for convex composite optimization by applying the standard \glsxtrfull{bpg} method to a mirror-coordinate reparameterization of the objective.  
  In primal variables, the resulting algorithm alternates between a preconditioned gradient step followed by a right Bregman proximal update.
  Our analysis relies on relative smoothness holding in the mirror coordinates instead of primal coordinates.
  Under this condition, we prove monotonic decrease of the objective in the general convex setting.
  We then specialize the method to Poisson inverse problems using weighted negative entropy as the potential.
  The resulting scheme recovers the classical Richardson--Lucy multiplicative updates and extends them to general convex regularizers.
  Building on a recent convergence analysis of multiplicative updates, we establish a sublinear convergence rate in function values for the Poisson setting.
  Finally, we demonstrate the performance of the regularized algorithm on several imaging inverse problems with Poisson-distributed observations.
  The code is publicly available at \url{https://github.com/Tmodrzyk/MU-Bregman}.
\end{abstract}



\section{Introduction}
\label{sec:introduction}

Poisson inverse problems arise in a wide range of applications, including astronomy, fluorescence microscopy \cite{berteroInverseImagingPoisson2018}, photon-limited imaging \cite{salmonPoissonNoiseReduction2014a} and emission tomography \cite{sheppMaximumLikelihoodReconstruction1982b}.
The goal is to recover $x \in \mathbb{R}^n_+$ from measurements $y \in \mathbb{R}^m_+$ modeled as:
\begin{equation}
    y \sim \mathcal{P}(Ax)
\end{equation}
where $\mathcal{P}$ is the Poisson distribution.
Depending on the problem at hand, the nonnegative matrix $A \in \mathbb{R}^{m \times n}_+$ can represent a convolution operation, the Radon transform, or any other linear image degradation.
The forward operator $A$ is often ill-conditioned, making reconstruction sensitive to measurement noise.
A standard variational approach to solve Poisson inverse problems consists in minimizing the Poisson data-fidelity $f: \mathbb{R}^n_+ \to \mathbb{R} \cup \left\{+\infty\right\}$:
\begin{equation}
  f(x) = \sum_{j=1}^m \left[ (Ax)_j - y_j \log (Ax)_j \right]
  \label{eq:poisson-nll}
\end{equation}
which coincides, up to constants independent of $x$, with the \glsxtrfull{nll} of the Poisson model.
It is also equal, up to additive constants, to the generalized \glsxtrfull{kl} between $y$ and $Ax$.
For nonnegative vectors $u,v\in\mathbb{R}^m_+$, its expression is given by:
\begin{equation*}
    \text{\glsxtrshort{kl}}(u\parallel v)
    = \sum_{i=1}^m \left[u_i\log\left(\frac{u_i}{v_i}\right)- (u_i - v_i)\right],
\end{equation*}
with the conventions $0\log(0)=0$ and $u_i\log(u_i/0)=+\infty$ for $u_i>0$.

The gradient of the Poisson \glsxtrshort{nll} is locally Lipschitz, but its Lipschitz constant explodes near the boundary of $\mathbb{R}^n_+$.
Moreover, for non-diagonal $A$, the proximal operator of $f$ generally has no closed form, limiting the use of proximal splitting methods.
Still, the Poisson data-fidelity can be minimized by the classical multiplicative updates:
\begin{equation}
  x^{(k+1)} = \frac{x^{(k)}}{s} \odot A^\top \frac{y}{Ax^{(k)}}
  \label{eq:multiplicative-updates}
\end{equation}
where $s = A^\top \mathbf{1}_m$ is the vector of column sums of $A$, $\odot$ denotes the element-wise product, and the division is also taken element-wise.
This algorithm first appeared in the 1970s in astronomical image deconvolution \cite{richardsonBayesianBasedIterativeMethod1972a,lucyIterativeTechniqueRectification1974} and in channel capacity estimation \cite{arimotoAlgorithmComputingCapacity1972,blahutComputationChannelCapacity1972a}.
It was later rediscovered in emission tomography as an instance of the \glsxtrfull{em} algorithm \cite{sheppMaximumLikelihoodReconstruction1982b,dempsterMaximumLikelihoodIncomplete1977}.
In the context of nonnegative matrix factorization \cite{NIPS2000_f9d11525}, it is usually interpreted as a \glsxtrfull{mm} method \cite{langeMMOptimizationAlgorithms2016}.
The multiplicative updates \eqref{eq:multiplicative-updates} are known to converge to a minimizer of the Poisson objective \cite{iusemShortConvergenceProof1992} while preserving the non-negativity constraints.
Additionally, a sublinear $\mathcal{O}(1/k)$ convergence rate in function values was recently established \cite{zhaoConvergenceRateAnalysis2023}.

In practice, however, the maximum likelihood solution tends to produce reconstructions containing artifacts \cite{snyderNoiseEdgeArtifacts1987}.
Early stopping is often used as an implicit regularization strategy, but this generally comes at the cost of reduced reconstruction quality.
Another approach is to introduce a regularization term $g : \mathbb{R}^n_+ \to \mathbb{R} \cup \{+\infty\}$ and solve the composite problem
\begin{equation}
    \min_{x \in \mathbb{R}^n_+} \left\{ f(x) + g(x) \right\}.
    \label{eq:composite-minimization}
\end{equation}
Several methods have been proposed for \eqref{eq:composite-minimization}, including regularized variants of the multiplicative algorithm derived from \glsxtrshort{em} or \glsxtrshort{mm} arguments \cite{ahnGloballyConvergentImage2003,deyRichardsonLucyAlgorithm2006}.
However, these methods are often tailored to specific choices of regularizer and are difficult to generalize.

More recently, Bregman methods have become popular for problems involving functions that are not necessarily $L$-smooth \cite{bauschkeDescentLemmaLipschitz2017,luRelativelySmoothConvex2018}.
These methods extend Euclidean gradient and proximal algorithms by replacing the squared norm with a divergence induced by a convex potential $\varphi$ defined as:
\begin{equation}
  D_\varphi(x \parallel y) = \varphi(x) - \varphi(y) - \langle \nabla \varphi(y), x - y \rangle.
  \label{eq:bregman-divergence}
\end{equation}
This relaxes the usual $L$-smoothness requirement on $f$ to smoothness relative to the geometry induced by $\varphi$.
This framework has been applied to Poisson inverse problems \cite{bauschkeDescentLemmaLipschitz2017}, but the resulting algorithms require a small descent step size, which can lead to slower convergence in practice.
For convolution operators, the Richardson--Lucy algorithm has been analyzed using mirror descent and relative smoothness \cite{aubin-frankowskiMirrorDescentRelative2022a}. 
The present work develops this perspective for the multiplicative updates \eqref{eq:multiplicative-updates} with a general nonnegative linear operator $A$, and extends them to convex regularization.

We formulate a \glsxtrfull{bpg} method in the mirror space defined by $\nabla \varphi(x)$ instead of applying it to the primal variables $x$ as is usually done.
Rewriting the resulting updates in primal variables yields the general scheme, with step size $\tau>0$:
\begin{equation}
  x^{(k+1)} = \rightprox^{\varphi}_{\tau g} \left( x^{(k)} - \tau\left[ \nabla^2 \varphi(x^{(k)}) \right]^{-1} \nabla f(x^{(k)}) \right),
  \label{eq:generalized-updates}
\end{equation}
where $\rightprox^{\varphi}_g$ is the right Bregman proximal operator of $g$ defined as:
\begin{equation}
  \rightprox^{\varphi}_g(x) = \argmin_{u \in \intdom\varphi} \left\{ g(u) + D_\varphi(x \parallel u)\right\}.
  \label{eq:right-bregman-prox}
\end{equation}
The updates \eqref{eq:generalized-updates} are alternating between a preconditioned gradient step reminiscent of Riemannian gradient descent and natural gradient descent \cite{amariInformationGeometryIts2016}, and a right Bregman proximal step.
The appearance of the \emph{right} Bregman proximal operator is unusual as recent developments are mostly based on the \emph{left} Bregman proximal operator.
We name this new scheme \glsxtrfull{rbpg} and propose a first analysis of its convergence properties, proving monotonic decrease of the objective function along the iterates.
Notably, during our analysis, a new smoothness condition arises, distinct from the usual relative smoothness: 
\begin{equation*}
  \nabla^2 f(x)-\nabla^3\varphi(x)\left[\left[\nabla^2\varphi(x)\right]^{-1}\nabla f(x)\right] \preceq L \nabla^2\varphi(x).
\end{equation*} 

We then specialize \glsxtrshort{rbpg} to the Poisson \glsxtrshort{nll} using the weighted Boltzmann--Shannon entropy, rather than the Burg entropy used in the standard relative smoothness approach.
This choice allows a unit step size and recovers the classical multiplicative updates in the unregularized case.
With a convex regularization term $g$, the method becomes
\begin{equation}
  x^{(k+1)} = \rightprox^{\varphi}_g \left( \frac{x^{(k)}}{s} \odot A^\top \frac{y}{Ax^{(k)}} \right).
  \label{eq:generalized-multiplicative-updates}
\end{equation}
We then combine the above Bregman interpretation with the recent convergence analysis of non-regularized multiplicative updates \cite{zhaoConvergenceRateAnalysis2023}, and recover a standard $\mathcal{O}(1/k)$ convergence rate in function values.

\begin{table}[!h]
  \caption{Summary of our contributions in the right column, compared with the existing literature.}
  \label{tab:gradient-method-comparison}
  \centering
  \renewcommand{\arraystretch}{1.8}
  \resizebox{\linewidth}{!}{%
  \begin{tabular}{@{}lccc@{}}
    \hline
    & PGD & BPG \cite{bauschkeDescentLemmaLipschitz2017,luRelativelySmoothConvex2018} & R-BPG (ours) \\
    \hline
    General update
      & $\operatorname{prox}_{\tau g}\!\left[x^{(k)}-\tau\nabla f(x^{(k)})\right]$
      & $\leftprox^\varphi_{\tau g}\!\left[\nabla\varphi^*\!\left(\nabla\varphi(x^{(k)})-\tau\nabla f(x^{(k)})\right)\right]$
      & $\rightprox^\varphi_{\tau g}\!\left[x^{(k)}-\tau[\nabla^2\varphi(x^{(k)})]^{-1}\nabla f(x^{(k)})\right]$ \\
    Smoothness condition
      & $\nabla^2 f(x)\preceq LI$
      & $\nabla^2 f(x)\preceq L\nabla^2\varphi(x)$
      & $\nabla^2 f(x)-\nabla^3\varphi(x)\left[\left[\nabla^2\varphi(x)\right]^{-1}\nabla f(x)\right]\preceq L\nabla^2\varphi(x)$ \\
    Poisson update
      & --
      & $\leftprox^\varphi_{\tau g}\!\left(\dfrac{x^{(k)}}{1+\tau x^{(k)}\odot\nabla f(x^{(k)})}\right)$
      & $\rightprox^\varphi_g\!\left(\dfrac{x^{(k)}}{s}\odot A^\top\dfrac{y}{Ax^{(k)}}\right)$ \\
    Poisson step size
      & --
      & $0<\tau\leq 1/\|y\|_1$
      & $0<\tau\leq 1$ \\
    \hline
  \end{tabular}}
\end{table}

\subsection{Outline of the paper} In section \ref{sec:relative-smoothness} we introduce the relative smoothness framework and its application to Poisson inverse problems.
In section \ref{sec:regularized-scheme} we motivate the introduction of a descent in reparameterized coordinates, and show how it leads to the general scheme \eqref{eq:generalized-updates}.
We then apply it to the Poisson data-fidelity and where it recover the multiplicative updates in the unregularized case, and yields the modified updates \eqref{eq:generalized-multiplicative-updates} in the regularized case.
In section \ref{sec:convergence-analysis-convex} we propose a first analysis of the general scheme \eqref{eq:generalized-updates} under a modified relative smoothness assumption, and later specialize the analysis to the Poisson case to establish convergence rates in function values.
In section \ref{sec:experiments} we conduct numerical experiments to compare the regularized multiplicative updates against existing methods for Poisson inverse problems.
The code is available on Github: \url{https://github.com/Tmodrzyk/MU-Bregman}.
Finally in section \ref{sec:discussions} we discuss the perspectives of this work.

\section{\texorpdfstring{\Glsxtrlong{bpg}}{Bregman proximal gradient} methods}
\label{sec:relative-smoothness}

In this section we present an overview of \glsxtrlong{bpg} methods and their analysis within the relative smoothness framework.
For a thorough introduction to this topic, we refer the reader to the seminal work \cite{bauschkeDescentLemmaLipschitz2017,luRelativelySmoothConvex2018} and to the review \cite{teboulleSimplifiedViewFirst2018}.
We then review how these notions relate to Poisson inverse problems.

\subsection{Relative smoothness}

Let us recall the composite optimization problem:
\begin{equation}
    \min_{x \in \mathbb{R}^n} \left\{ f(x) + g(x) \right\}
    \label{eq:composite-optimization-convex}
\end{equation}
where $f: \mathbb{R}^n \to \mathbb{R} \cup \{+\infty\}$ is a proper, \glsxtrfull{lsc} and differentiable function, while $g: \mathbb{R}^n \to \mathbb{R} \cup \{+\infty\}$ is proper, \glsxtrshort{lsc}, convex, but possibly non-differentiable.
We denote $\Phi = f + g$ the composite objective function.
In order to use a gradient-based method to minimize \eqref{eq:composite-optimization-convex}, it is standard to assume that $f$ is $L$-smooth, meaning that its gradient is Lipschitz continuous with constant $L > 0$.
For convex functions, this property is equivalent to being able to upper bound $f$ by a quadratic function at any point, a property that is often called the descent lemma:
\begin{equation}
  \label{eq:L-smoothness-upper-bound}
  f(u) \leq f(x) + \langle \nabla f(x), u - x \rangle + \frac{L}{2} \|u - x\|^2, \quad \forall u, x \in \mathbb{R}^n.
\end{equation}

This Euclidean notion of smoothness is however too restrictive for many objectives of interest.
In particular, functions such as the Poisson \glsxtrshort{nll} in \eqref{eq:poisson-nll} do not satisfy such a quadratic upper bound, and therefore are not $L$-smooth in the classical sense.
The idea of relative smoothness is to use a different geometry described by the Bregman divergence $D_\varphi$ (see equation \eqref{eq:bregman-divergence}) associated to a potential $\varphi: \mathbb{R}^n \to \mathbb{R} \cup \{+\infty\}$.
This new notion leads to a different upper bound:
\begin{equation}
  \label{lem:descent}
  f(u) \leq f(x) + \langle \nabla f(x), u - x \rangle + L D_\varphi(u \parallel x), \quad \forall u, x \in \intdom \varphi,
\end{equation}
where $L > 0$ is a Lipschitz like constant.
To induce a well-behaved Bregman divergence, the potential function $\varphi$ is required to be of Legendre type :
\begin{definition}[Legendre-type function \thmcite{bauschkeDescentLemmaLipschitz2017}{definition~1}]
  \label{def:legendre}
  A proper, \glsxtrshort{lsc} convex function $\varphi: \mathbb{R}^n \to \mathbb{R} \cup \{+\infty\}$ is of Legendre type if it satisfies both conditions:
  \begin{enumerate}
    \item \emph{Essential strict convexity:} $\varphi$ is strictly convex on every convex subset of $\operatorname{dom}\partial\varphi$.
    \item \emph{Essential smoothness:} $\intdom\varphi\neq\emptyset$, $\varphi$ is differentiable on $\intdom\varphi$, and for every sequence $(x^{(k)})_{k\in\mathbb{N}}\subset\intdom\varphi$ converging to a point in $\partial\operatorname{dom}\varphi$,
    \[
      \|\nabla\varphi(x^{(k)})\|\longrightarrow+\infty.
    \]
  \end{enumerate}
\end{definition}
Let us recall that the Fenchel conjugate of a function $\varphi$ is defined as:
\begin{equation}
  \varphi^*(\eta) = \sup_{x \in \mathbb{R}^n} \left\{ \langle \eta, x \rangle - \varphi(x) \right\}.
\end{equation}
For Legendre-type functions there is a strong link between $\varphi$ and its Fenchel conjugate $\varphi^*$:
\begin{lemma}[{\cite[theorem 26.5]{rockafellarConvexAnalysis1997}}]
  \label{lem:legendre-conjugacy}
  A proper, \glsxtrshort{lsc} convex function $\varphi: \mathbb{R}^n \to \mathbb{R} \cup \{+\infty\}$ is of Legendre type if and only if its Fenchel conjugate $\varphi^*$ is of Legendre type.
\end{lemma}
Legendre functions and their associated Bregman divergences enjoy interesting properties that we will use throughout the rest of the paper and that we recall here.
For a direction $r$, we note the third order derivative as:
\begin{equation*}
  \left[\nabla^3 \varphi(x)[r]\right]_{ij}
  = \sum_{k=1}^n \frac{\partial^3 \varphi(x)}{\partial x_i\,\partial x_j\,\partial x_k} r_k.
\end{equation*}
\begin{proposition}[Properties of Legendre-type functions]
  \label{prop:legendre-properties}
  Consider $\varphi: \mathbb{R}^n \to \mathbb{R} \cup \{+\infty\}$ a Legendre-type function:
  \begin{enumerate}
    \item $\nabla \varphi: \intdom \varphi \to \intdom \varphi^*$ is a bijection, and its inverse is $\left( \nabla \varphi \right)^{-1} = \nabla \varphi^*$.
    \item For all $x, y \in \intdom \varphi$, we have:
    \begin{equation}
        D_\varphi(x \parallel y) = D_{\varphi^*}(\nabla \varphi(y) \parallel \nabla \varphi(x)),
    \end{equation}
    where the order of the arguments is reversed in the right-hand side.
    \item If $\varphi$ is $\mathcal C^2$ on $\intdom\varphi$ with $\nabla^2\varphi\succ0$, then for all $x \in \intdom \varphi$:
    \begin{equation}
      \left[ \nabla^2 \varphi(x) \right]^{-1} = \nabla^2 \varphi^*(\nabla \varphi(x)) 
    \end{equation}
    \item If $\varphi$ is $\mathcal C^3$ on $\intdom\varphi$ with $\nabla^2\varphi\succ0$, then for all $x\in\intdom\varphi$ and $r\in\mathbb R^n$,
    \begin{equation}
      \label{eq:legendre-third-order}
      \nabla^3 \varphi^*(\nabla \varphi(x)) \left[r\right]
      = - \left[ \nabla^2 \varphi(x) \right]^{-1} \nabla^3 \varphi(x) \left[ \left[ \nabla^2 \varphi(x) \right]^{-1} r \right] \left[ \nabla^2 \varphi(x) \right]^{-1}.
    \end{equation}
  \end{enumerate}
\end{proposition}
\begin{proof}
  A proof of the first point can be found in \cite[theorem 26.5]{rockafellarConvexAnalysis1997}.
  The second point is straightforward to verify by direct computation. 
  For the third point, we start from the first identity:
  \begin{equation*}
    \nabla \varphi^*(\nabla \varphi(x)) = x,
  \end{equation*}
  and differentiate with respect to $x$, which gives by the chain rule
  \begin{equation*}
    \nabla^2\varphi^*(\nabla\varphi(x))\nabla^2\varphi(x)=I.
  \end{equation*}
  Since $\nabla^2\varphi(x)\succ0$, we can multiply by its inverse on the right to obtain the stated identity.
  For the fourth point, we differentiate the product identity $\nabla^2\varphi^*(\nabla\varphi(x))\nabla^2\varphi(x)=I$ with respect to $x$ in a direction $u\in\mathbb R^n$.
  The product rule and the chain rule give
  \begin{equation*}
    \nabla^3\varphi^*(\nabla\varphi(x))[\nabla^2\varphi(x)u]\nabla^2\varphi(x) + \nabla^2\varphi^*(\nabla\varphi(x))\nabla^3\varphi(x)[u]=0.
  \end{equation*}
  Multiplying on the right by $\left[\nabla^2\varphi(x)\right]^{-1}$ and using the third point gives
  \begin{equation*}
    \nabla^3\varphi^*(\nabla\varphi(x))[\nabla^2\varphi(x)u]
    =-\left[\nabla^2\varphi(x)\right]^{-1}\nabla^3\varphi(x)[u]\left[\nabla^2\varphi(x)\right]^{-1}.
  \end{equation*}
  Taking $u =\left[\nabla^2\varphi(x)\right]^{-1}r$ gives \eqref{eq:legendre-third-order}.
\end{proof}

A notable example of Legendre function is the energy potential $\varphi(x) = \frac{1}{2} \|x\|^2$, whose Bregman divergence is the squared Euclidean distance.
Another important example is the negative entropy, sometimes called Boltzmann or Shannon entropy, $\varphi(x) = \sum_{i=1}^n \left[ x_i \log x_i - x_i \right]$, which yields the Kullback--Leibler divergence.

With $\varphi$ of Legendre type, the divergence $D_\varphi$ is non-negative for any $x, y \in \intdom \varphi$, and $D_\varphi(x \parallel y) = 0$ if and only if $x = y$.
However, note that it is not a true distance, as it is not symmetric and does not satisfy the triangle inequality in general.
It is also convex with respect to its first argument, but not necessarily with respect to its second one.
It can be proven by simple computations that the $L$-relative smoothness upper bound \eqref{lem:descent} can equivalently be expressed by the convexity of the function $L \varphi - f$ on $\intdom \varphi$:
\begin{definition}[Relative smoothness \cite{bauschkeDescentLemmaLipschitz2017}]
  \label{def:relative-smoothness}
  A proper, \glsxtrshort{lsc} function $f: \mathbb{R}^n \to \mathbb{R} \cup \{+\infty\}$ is said to be $L$-smooth relative to a Legendre-type function $\varphi: \mathbb{R}^n \to \mathbb{R} \cup \{+\infty\}$ if there exists $L > 0$ such that $L \varphi - f$ is convex on $\intdom \varphi$.
  If $f$ is differentiable on $\intdom\varphi$ and is $L$-relatively smooth, it verifies \eqref{lem:descent}.
  Equivalently, when both $f$ and $\varphi$ are twice differentiable:
  \begin{equation}
    L \nabla^2 \varphi(x) \succeq \nabla^2 f(x), \quad \forall x \in \intdom \varphi.
  \end{equation}
\end{definition}

This definition does not require $f$ to be convex.
We note that several variations of the definition of relative smoothness exist in the literature (see for instance \cite[definition 2.2]{bolteFirstOrderMethods2018}), but the one above is sufficient for our purposes.

\subsection{Bregman proximal gradient algorithm}
\label{sec:relative-smoothness:subsec:bregman-proximal}

Assuming that we can find a Legendre type potential $\varphi$ such that $f$ is $L$-smooth relative to $\varphi$, an algorithm to minimize $f$ can be derived from inequality \eqref{lem:descent}.
The iterates are solution of the following minimization problem:
\begin{equation}
  x^{(k+1)} = \argmin_{x \in \operatorname{dom} \varphi} \left\{ f(x^{(k)}) + \langle \nabla f(x^{(k)}), x - x^{(k)} \rangle + \frac{1}{\tau} D_\varphi(x \parallel x^{(k)}) \right\}
  \label{eq:mirror-descent-minimization}
\end{equation}
for a descent step size $\tau > 0$.
Writing the first-order optimality condition of this minimization problem gives:
\begin{equation*}
  \nabla \varphi(x^{(k+1)}) = \nabla \varphi(x^{(k)}) - \tau \nabla f(x^{(k)})
\end{equation*}
and using the fact that $\nabla \varphi^* = (\nabla \varphi)^{-1}$ (see point 1 of proposition \ref{prop:legendre-properties}), we can rewrite the update as:
\begin{equation}
  x^{(k+1)} = \nabla \varphi^* \left(\nabla \varphi (x^{(k)}) - \tau \nabla f(x^{(k)})\right),
  \label{eq:mirror-descent}
\end{equation}
which recovers the well-known mirror descent algorithm \cite{nemirovskiiProblemComplexityMethod1983, beckMirrorDescentNonlinear2003a}.

A similar approach can be adopted for the composite minimization problem \eqref{eq:composite-minimization}.
Similar to the Euclidean case, we can define a notion of proximal operator, with respect to the potential $\varphi$.
Replacing the squared Euclidean distance with a Bregman divergence in the definition of the proximal operator, we can define the \emph{left} Bregman proximal operator of $g$ as:
\begin{equation}
  \leftprox^{\varphi}_g(x) = \argmin_{u \in \operatorname{dom} \varphi} \left\{ g(u) + D_\varphi(u \parallel x)\right\}
  \label{eq:left-bregman-prox}
\end{equation}
The asymmetry of the Bregman divergence enables to define a \emph{right} Bregman proximal operator $\rightprox^{\varphi}_g$ where the optimized variable $u$ appears in the second argument of the divergence, as defined in \eqref{eq:right-bregman-prox}.
This right Bregman proximal operator has been used in the context of alternating left and right Bregman proximal algorithms \cite{bauschkeJointMinimizationAlternating2006}, will naturally arise in our derivations in section \ref{sec:regularized-scheme}.

We can slightly modify the definition of the mirror descent updates in \eqref{eq:mirror-descent-minimization} to handle the composite optimization problem \eqref{eq:composite-optimization-convex} by adding the regularization term $g$ to the minimization problem \eqref{eq:mirror-descent-minimization}:
\begin{equation}
  x^{(k+1)} = \argmin_{x \in \operatorname{dom} \varphi} \left\{ f(x^{(k)}) + \langle \nabla f(x^{(k)}), x - x^{(k)} \rangle + \frac{1}{\tau} D_\varphi(x \parallel x^{(k)}) + g(x) \right\}.
  \label{eq:bregman-proximal-descent-minimization}
\end{equation}
Developing the expression of the Bregman divergence \eqref{eq:bregman-divergence}, the minimization problem above can equivalently be rewritten as:
\begin{equation*}
  \argmin_{x \in \operatorname{dom} \varphi} \left\{ \tau g(x) + \varphi(x) - \langle \nabla \varphi(x^{(k)}) - \tau \nabla f(x^{(k)}), x - x^{(k)} \rangle \right\},
\end{equation*}
defining the intermediate point $z^{(k)} = \nabla \varphi^* \left( \nabla \varphi(x^{(k)}) - \tau \nabla f(x^{(k)}) \right)$:
\begin{equation*}
  \argmin_{x \in \operatorname{dom} \varphi} \left\{ \tau g(x) + \varphi(x) - \langle \nabla \varphi( z^{(k)}), x \rangle \right\},
\end{equation*}
which is the left Bregman proximal operator of $g$ evaluated at $z^{(k)}$:
\begin{equation*}
  \argmin_{x \in \operatorname{dom} \varphi} \left\{ \tau g(x) + D_\varphi(x \parallel z^{(k)}) \right\} = \leftprox^{\varphi}_{\tau g}(z^{(k)}).
\end{equation*}
The updates \eqref{eq:bregman-proximal-descent-minimization} thus alternate a mirror descent and left Bregman proximal steps:
\begin{equation}
  x^{(k+1)} = \leftprox^{\varphi}_{\tau g} \left[ \nabla \varphi^* \left(\nabla \varphi (x^{(k)}) - \tau \nabla f(x^{(k)})\right) \right]
  \label{eq:bregman-proximal-descent}
\end{equation}
and are known as \glsxtrlong{bpg} descent \cite{bauschkeDescentLemmaLipschitz2017}.

For convex $f$ and $g$, under the relative smoothness property \ref{def:relative-smoothness}, the \glsxtrshort{bpg} algorithm enjoys convergence guarantees similar to the usual Euclidean proximal gradient descent algorithm.
The descent inequality \eqref{lem:descent} can be used in combination with properties of the left Bregman operator to show that \glsxtrshort{bpg} converges in $\mathcal{O}(1/k)$ in function values to a minimum of $f + g$ in  \cite[theorem 1]{bauschkeDescentLemmaLipschitz2017}, and under additional assumptions satisfied by standard choices of Bregman divergences, full-sequence convergence of the iterates to a solution of the composite problem \eqref{eq:composite-optimization-convex} was shown \cite[theorem 2]{bauschkeDescentLemmaLipschitz2017}.

\subsection{Application to Poisson inverse problems}
\label{sec:relative-smoothness:subsec:poisson}

The Poisson \glsxtrshort{nll} defined in \eqref{eq:poisson-nll} is not $L$-smooth, which makes it a good candidate for the relative smoothness framework.
The main challenge is then to find a Legendre-type potential $\varphi$ such that the function is $L$-smooth relative to $\varphi$.
\begin{lemma}[Poisson \glsxtrshort{nll} and Burg's entropy \thmcite{bauschkeDescentLemmaLipschitz2017}{lemma 7}]
  \label{lem:relative-smoothness-poisson}
  The Poisson \glsxtrshort{nll} defined in \eqref{eq:poisson-nll} is $L$-smooth relative to Burg's entropy $\varphi(x) = - \sum_{i=1}^n \log x_i$ with $L = \| y \|_1$.
\end{lemma}
This property yields the \glsxtrshort{bpg} algorithm for the Poisson \glsxtrshort{nll} with Burg's entropy as potential \cite{bauschkeDescentLemmaLipschitz2017}:
\begin{equation}
  x^{(k+1)} = \leftprox^\varphi_{\tau g} \left[ \frac{x^{(k)}}{1 + \tau x^{(k)} \odot \nabla f(x^{(k)})} \right]
  \label{eq:bpg-poisson}
\end{equation}
and which is convergent for a suitable step size $0< \tau \leq 1 / \| y \|_1$.

When $g=0$, these updates reduce to mirror descent with Burg's entropy.
They differ from the multiplicative updates: no constant step size $\tau$ yields the iterates \eqref{eq:multiplicative-updates}.
In practice $\| y \|_1$ can be large, especially in emission tomography where the data is 3D, and the dimension of the sinogram can exceed $10^7$ elements.
This would typically lead to small step sizes and slow convergence of BPG in practice if this conservative bound was used.
Backtracking or line-search procedures allow the step size to adapt at each iteration.
In our experiments, we use Armijo backtracking for Burg mirror descent and BPG, starting each iteration from a unit trial step size.

Another interesting choice of potential is the negative entropy $\varphi(x) = \sum_{i=1}^n \left[ x_i \log x_i - x_i \right]$, which yields the Kullback--Leibler divergence.
Indeed, the Poisson data-fidelity $f$ is equal, up to additive constants, to the \glsxtrshort{kl} divergence between $y$ and $Ax$.
One could then imagine that the geometry induced by the negative entropy could be related to the Poisson data-fidelity.
Unfortunately, the Poisson \glsxtrshort{nll} is not relatively smooth with respect to the negative entropy, as we show in appendix \ref{proof:lem:non-relative-smoothness-poisson}.
The mirror descent updates associated to the negative entropy are:
\begin{equation}
  x^{(k+1)} = x^{(k)} \odot \exp \left( - \nabla f(x^{(k)})\right),
\end{equation}
which is sometimes known as \emph{entropic} mirror descent. 
In the Poisson setting, it has been used in emission tomography \cite{bernardSystemeModulaireTraitement2017} where it appeared to have approximately similar behavior as the regular multiplicative updates \eqref{eq:multiplicative-updates}.
When the mirror descent step is constrained to the unit simplex $\Delta_n = \{ x \in \mathbb{R}^n_+ : \sum_{i=1}^n x_i = 1\}$, the updates are instead:
\begin{equation}
  x^{(k+1)} = \frac{x^{(k)} \odot \exp \left( - \nabla f(x^{(k)}) \right)}{\sum_{i=1}^n x_i^{(k)} \exp \left( - [\nabla f(x^{(k)})]_i \right)},
  \label{eq:exponentiated-gradient}
\end{equation}
and are sometimes called \emph{exponentiated gradient} \cite{kivinenExponentiatedGradientGradient1997} or more generally \emph{multiplicative weight updates} in game theory \cite{aroraMultiplicativeWeightsUpdate2012b}
This instance of mirror descent has been studied in the context of optimization in \cite{beckMirrorDescentNonlinear2003a}, and has been shown to converge under the usual $L$-smoothness assumption on $f$, excluding the Poisson \glsxtrshort{nll} from the analysis.

While the Poisson data-fidelity seems incompatible with the negative entropy, the reverse objective $\text{\glsxtrshort{kl}}(Ax \parallel y)$ is $L$-smooth relative to the negative entropy \cite[lemma 8]{bauschkeDescentLemmaLipschitz2017} with constant $L = \| A \|_1$.
However, this reverse objective does not coincide with the negative log-likelihood of the Poisson observation model.

\section{Right Bregman proximal gradient algorithm}
\label{sec:regularized-scheme}

In this section, we introduce \glsxtrshort{rbpg}, a new algorithm for solving composite optimization problems of the form \eqref{eq:composite-optimization-convex}.
The algorithm is based on the application of \glsxtrshort{bpg} to a reparameterization of the objective function.
We first motivate \glsxtrshort{rbpg} with recent links established between mirror descent and classical algorithms for Poisson inverse problems.
We then derive the method in a general setting, before specializing it to the Poisson case.
The properties of this algorithm and the precise assumptions under which it is guaranteed to converge are discussed later in section \ref{sec:convergence-analysis-convex}.

\subsection{Motivations}
\label{subsec:prior-motivation}

Recent work has established a strong connection between the \glsxtrshort{em} algorithm \cite{dempsterMaximumLikelihoodIncomplete1977} and mirror descent \cite{kunstner2021homeomorphic, aubin-frankowskiMirrorDescentRelative2022a}.
In particular, in the parametric setting of exponential-family models, it was shown in \cite{kunstner2021homeomorphic} that \glsxtrshort{em} can be interpreted as a mirror descent scheme, with a mirror map $\nabla \varphi$ induced by the log-partition function of the complete-data likelihood written in natural parameters.
This perspective was later extended in \cite{aubin-frankowskiMirrorDescentRelative2022a}, which developed mirror descent and relative smoothness in spaces of measures.
Notably, their analysis made it possible to establish this connection between \glsxtrshort{em} and mirror descent in the nonparametric setting.
Among their applications, they proved that the Richardson-Lucy algorithm \cite{richardsonBayesianBasedIterativeMethod1972a,lucyIterativeTechniqueRectification1974}, corresponding to the multiplicative updates \eqref{eq:multiplicative-updates} where $A$ specifically models a convolution operator, can also be interpreted as a mirror descent method, and they provide sublinear rates in function values.
The general multiplicative updates \eqref{eq:multiplicative-updates} are also known to arise from an \glsxtrshort{em} framework, based on a suitable choice of latent variables \cite{sheppMaximumLikelihoodReconstruction1982b}.
These connections between \glsxtrshort{em} and mirror descent motivated us to revisit multiplicative updates from an optimization point of view.

\subsection{Formal derivation of the R-BPG updates}
\label{subsec:reparameterized-descent}

We consider a reparameterization of the objective induced by a Legendre potential $\varphi$, assumed to be $\mathcal{C}^2$ on $\intdom\varphi$ with $\nabla^2\varphi(x)\succ0$ for every $x\in\intdom\varphi$.
We study the function $f$ through the coordinates $\eta = \nabla \varphi(x)$: the reparameterized objective function is then $f \circ \nabla \varphi^*$.
Choosing $\varphi^*$ as potential, the mirror descent updates \eqref{eq:mirror-descent-minimization} in mirror space are then:
\begin{equation}
  \eta^{(k+1)} = \argmin_{\eta \in \intdom\varphi^*} \left\{ (f \circ \nabla \varphi^*)(\eta^{(k)}) + \langle \nabla (f \circ \nabla \varphi^*)(\eta^{(k)}),\eta - \eta^{(k)} \rangle + \frac{1}{\tau} D_{\varphi^*}(\eta \parallel \eta^{(k)}) \right\}
\end{equation}
For now, we assume that this mapping is well-defined, i.e. that the minimization problem above has a unique solution, and we will provide conditions for this to hold in section \ref{sec:convergence-analysis-convex}.
The first-order optimality condition of this minimization problem is:
\begin{equation*}
  \label{eq:first-order-reparameterized}
  \nabla \varphi^*(\eta^{(k+1)}) = \nabla \varphi^*(\eta^{(k)}) - \tau \nabla (f \circ \nabla \varphi^*)(\eta^{(k)}).
\end{equation*}
Using the chain rule, the reparameterized gradient can be written as:
\begin{equation}
  \label{eq:reparameterized-gradient}
  \nabla (f \circ \nabla \varphi^*)(\eta) = \nabla^2 \varphi^*(\eta) \nabla f(\nabla \varphi^*(\eta))
\end{equation}
and using point 3 of proposition \ref{prop:legendre-properties}, we have
\begin{equation}
    \nabla^2 \varphi^*(\eta) = \left[\nabla^2 \varphi(\nabla \varphi^*(\eta))\right]^{-1}.
\end{equation}
Hence we can rewrite \eqref{eq:first-order-reparameterized} as:
\begin{equation*}
  \nabla \varphi^*(\eta^{(k+1)}) = \nabla \varphi^*(\eta^{(k)}) - \tau \left[ \nabla^2 \varphi(\nabla \varphi^*(\eta^{(k)})) \right]^{-1} \nabla f(\nabla \varphi^*(\eta^{(k)})).
\end{equation*}
Finally, using the first point of proposition \ref{prop:legendre-properties}, the update becomes in primal coordinates:
\begin{equation}
  x^{(k+1)} = x^{(k)} - \tau \left[ \nabla^2 \varphi(x^{(k)}) \right]^{-1} \nabla f(x^{(k)}).
  \label{eq:preconditioned-gradient-descent}
\end{equation}

Thus, mirror descent on the reparameterized objective in mirror coordinates reduces, in primal coordinates, to a preconditioned gradient step where the preconditioner is the inverse Hessian of $\varphi$.
This update is reminiscent of Riemannian gradient descent and, more specifically, of natural gradient descent \cite{amariInformationGeometryIts2016}, in which the metric tensor is given by the Fisher information matrix.
This connection is well-known and has been notably highlighted in the exponential-family setting \cite{raskuttiInformationGeometryMirror2015a} as well as in the continuous time limit \cite{gunasekarMirrorlessMirrorDescent2021b}.

We now extend this construction to the composite objective $\Phi = f + g$.
We introduce the notation $z^{(k)} =  x^{(k)} - \tau \left[ \nabla^2 \varphi(x^{(k)}) \right]^{-1} \nabla f(x^{(k)})$.
The \glsxtrshort{bpg} mapping applied to the reparameterized objective $\Phi \circ \nabla \varphi^*$ is, ignoring constants independent of $\eta$:
\begin{equation}
  \eta^{(k+1)} = \argmin_{\eta \in \intdom\varphi^*} \left\{ \langle \nabla (f \circ \nabla \varphi^*)(\eta^{(k)}), \eta - \eta^{(k)} \rangle + \frac{1}{\tau} D_{\varphi^*}(\eta \parallel \eta^{(k)}) + g \circ \nabla \varphi^*(\eta) \right\}.
  \label{eq:right-bregman-proximal-gradient-minimization-eta}
\end{equation}
Similarly to the derivations done in section \ref{sec:relative-smoothness:subsec:bregman-proximal}, we rewrite the above updates as:
\begin{equation*}
  \eta^{(k+1)} = \leftprox^{\varphi^*}_{\tau g \circ \nabla \varphi^*}  \circ \nabla \varphi \circ  \left( \nabla \varphi^* - \tau \nabla (f \circ \nabla \varphi^*)\right) (\eta^{(k)}).
  \label{eq:update-ugly}
\end{equation*}
We can significantly simplify these updates and rewrite them in the primal coordinates. 
In order to do so, we make explicit the relationship between the Bregman proximal operators in primal and mirror coordinates.
\begin{lemma}[Link between left and right Bregman proximal operators]
\label{lem:left-right-prox}
Let $\varphi: \mathbb{R}^n \to \mathbb R\cup\{+\infty\}$ be a Legendre-type function, and $g: \mathbb{R}^n \to \mathbb R\cup\{+\infty\}$ be a proper \glsxtrshort{lsc} function.
Assume that all the proximal problems appearing below have unique minimizers lying in the interiors of the domains of their respective potentials.
Then in $\intdom \varphi^*$ for $\tau > 0$:
\begin{equation}
  \label{eq:left-to-right}
  \leftprox^{\varphi^*}_{\tau g \circ \nabla \varphi^*} = \nabla \varphi \circ \rightprox^{\varphi}_{\tau g} \circ \nabla \varphi^*,
\end{equation}
and symmetrically in $\intdom \varphi$:
\begin{equation}
  \label{eq:right-to-left}
  \rightprox^{\varphi}_{\tau g} = \nabla \varphi^* \circ \leftprox^{\varphi^*}_{\tau g \circ \nabla \varphi^*} \circ \nabla \varphi.
\end{equation}
Similarly the right Bregman proximal operator of $g$ in mirror space is tied to the left Bregman proximal operator in primal space:
\begin{equation}
  \label{eq:right-dual-to-left-primal}
  \rightprox^{\varphi^*}_{\tau g \circ \nabla \varphi^*}
  =  \nabla \varphi \circ \leftprox^{\varphi}_{\tau g} \circ \nabla \varphi^*,
\end{equation}
or, equivalently,
\begin{equation}
  \label{eq:left-primal-to-right-dual}
  \leftprox^{\varphi}_{\tau g}
  = \nabla \varphi^* \circ \rightprox^{\varphi^*}_{\tau g \circ \nabla \varphi^*} \circ \nabla \varphi.
\end{equation}
\end{lemma}
\begin{proof}
  Let us first write
  \begin{equation*}
    \left[ \rightprox^{\varphi}_{\tau g} \circ \nabla \varphi^* \right] (\eta) = \argmin_{u \in \intdom \varphi} \left\{ g(u) + \frac{1}{\tau} D_\varphi(\nabla \varphi^* (\eta) \parallel u) \right\},
  \end{equation*}
  then with the change of variable $u = \nabla \varphi^*(\eta')$ to rewrite the above problem as:
  \begin{equation*}
    \left[ \rightprox^{\varphi}_{\tau g} \circ \nabla \varphi^* \right] (\eta) = \nabla \varphi^* \left( \argmin_{\eta' \in \intdom \varphi^*} \left\{ g \circ \nabla \varphi^*(\eta') + \frac{1}{\tau} D_\varphi(\nabla \varphi^* (\eta) \parallel \nabla \varphi^*(\eta')) \right\} \right).
  \end{equation*}
  We then use point 2 of proposition \ref{prop:legendre-properties} to rewrite this as:
  \begin{align*}
    \left[ \rightprox^{\varphi}_{\tau g} \circ \nabla \varphi^* \right] (\eta) &= \nabla \varphi^* \left( \argmin_{\eta' \in \intdom \varphi^*} \left\{ g \circ \nabla \varphi^*(\eta') + \frac{1}{\tau} D_{\varphi^*}(\eta' \parallel \eta) \right\} \right) \\
    &= \left[ \nabla \varphi^* \circ \leftprox^{\varphi^*}_{\tau g \circ \nabla \varphi^*} \right] (\eta).
  \end{align*}
  Composing by $\nabla \varphi$ on the left gives the first relation, while the second is obtained by composing by $\nabla \varphi$ on the right. For the remaining relations, let $x \in \intdom \varphi$ and change variables as above:
  \begin{align*}
    \left[ \rightprox^{\varphi^*}_{\tau g \circ \nabla \varphi^*} \circ \nabla \varphi \right](x)
    &= \argmin_{\eta' \in \intdom \varphi^*}
      \left\{ g \circ \nabla \varphi^*(\eta')
      + \frac{1}{\tau}D_{\varphi^*}(\nabla \varphi(x) \parallel \eta') \right\} \\
    &= \nabla \varphi \left(
      \argmin_{u \in \intdom \varphi}
      \left\{ g(u) + \frac{1}{\tau}D_\varphi(u \parallel x) \right\}
      \right) \\
    &= \left[ \nabla \varphi \circ \leftprox^\varphi_{\tau g} \right](x),
  \end{align*}
  where the second equality follows from
  $D_{\varphi^*}(\nabla \varphi(x) \parallel \nabla \varphi(u)) = D_\varphi(u \parallel x)$.
  Composing this identity with the inverse mirror maps gives \eqref{eq:right-dual-to-left-primal} and \eqref{eq:left-primal-to-right-dual}.
\end{proof}

Using \eqref{eq:left-to-right} we can rewrite equation \eqref{eq:update-ugly} as:
\begin{equation}
  \eta^{(k+1)} = \nabla \varphi \circ \rightprox^{\varphi}_{\tau g} \circ \left( \nabla \varphi^* - \tau \nabla (f \circ \nabla \varphi^*)\right) (\eta^{(k)}).
\end{equation}
As previously detailed, the point $\nabla \varphi^* - \tau \nabla (f \circ \nabla \varphi^*)(\eta^{(k)})$ can be rewritten as \eqref{eq:preconditioned-gradient-descent} in primal coordinates.
Per the first point of proposition \ref{prop:legendre-properties}, $\nabla \varphi$ is a bijection, hence we can rewrite the complete updates in the primal coordinates $x$ to obtain:
\begin{equation}
  x^{(k+1)} = \rightprox^{\varphi}_{\tau g} \left( x^{(k)} - \tau \left[ \nabla^2 \varphi(x^{(k)}) \right]^{-1} \nabla f(x^{(k)}) \right).
  \label{eq:bpg-right}
\end{equation}
Interestingly, this algorithm involves the \emph{right} Bregman proximal operator, which is less commonly used in optimization than its \emph{left} counterpart.
Some of its properties were studied in \cite{bauschkeJointMinimizationAlternating2006}, in the context of schemes alternating left and right Bregman proximal steps.
To the best of our knowledge, however, the scheme \eqref{eq:bpg-right} itself is new.

\subsection{Mirror space relative smoothness}

For the algorithm to be well-behaved, the relative smoothness of $f \circ \nabla \varphi^*$ with respect to $\varphi^*$ appears as a natural condition.
However, this condition deals with functions living in the mirror space, comparing their curvature there, which makes comparison with the notion of relative smoothness in primal space difficult.
We propose a characterization of the relative smoothness condition in mirror space relying only on quantities measured in primal space:
\begin{proposition}[Mirror space relative smoothness characterization]
  \label{prop:mirror-relative-smoothness-characterization}
  Let $f: \mathbb{R}^n \to \Rinf$ be proper, \glsxtrshort{lsc}, and $\varphi: \mathbb{R}^n \to \Rinf$ be of Legendre type.
  Further assume that $f \in \mathcal{C}^2$ on $\intdom f$ and $\varphi \in \mathcal{C}^3$ on $\intdom \varphi$, with $\intdom \varphi \subset \intdom f$ and $\nabla^2\varphi(x)\succ0$ for every $x\in\intdom\varphi$. Then the following statements are equivalent:
  \begin{itemize}
    \item[(i)] $f \circ \nabla \varphi^*$ is $L$-smooth relative to $\varphi^*$ 
    \item[(ii)] $\forall x \in \intdom \varphi$:
    \begin{equation}
      \label{eq:mirror-relative-smoothness-characterization}
      \nabla^2 f(x) - \nabla^3 \varphi(x) \left[ \left[ \nabla^2 \varphi(x) \right]^{-1} \nabla f(x) \right] \preceq L \nabla^2 \varphi(x)
    \end{equation}
  \end{itemize}
\end{proposition}
\begin{proof}
  Given both $f$ and $\varphi$ are $\mathcal{C}^2$, the first statement is equivalent to: 
  \begin{equation}
    \label{eq:dual-coordinate-relative-smoothness-1}
    \forall \eta \in \intdom \varphi^*, ~ \nabla^2 (f \circ \nabla \varphi^*)(\eta) \preceq L \nabla^2 \varphi^*(\eta).
  \end{equation}
  Recall that: 
  \begin{equation*}
    \nabla (f \circ \nabla \varphi^*) (\eta) = \nabla^2 \varphi^*(\eta) \nabla f(\nabla \varphi^*(\eta))
  \end{equation*}
  We can now differentiate this expression again using the chain rule, which gives:
  \begin{equation*}
    \nabla^2 (f \circ \nabla \varphi^*) (\eta) = \nabla^3 \varphi^* (\eta) \left[ \nabla f(\nabla \varphi^* (\eta))\right] + \nabla^2 \varphi^* (\eta) \nabla^2 f( \nabla \varphi^*(\eta)) \nabla^2 \varphi^* (\eta).
  \end{equation*}
  We now rewrite everything in primal variable $x=\nabla \varphi^* (\eta)$, which gives:
  \begin{equation*}
    \nabla^2 (f \circ \nabla \varphi^*) (\eta) = \nabla^3 \varphi^* (\nabla \varphi(x)) \left[ \nabla f(x)\right] + \nabla^2 \varphi^* (\nabla \varphi(x)) \nabla^2 f(x) \nabla^2 \varphi^* (\nabla \varphi(x)).
  \end{equation*}
  The third point of proposition \ref{prop:legendre-properties} gives $\nabla^2 \varphi^* (\nabla \varphi(x)) = \left[ \nabla^2 \varphi(x)\right]^{-1}$. 
  Replacing in the previous equation yields:
  \begin{equation*}
    \nabla^2 (f \circ \nabla \varphi^*) (\eta) = \nabla^3 \varphi^* (\nabla \varphi(x)) \left[ \nabla f(x)\right] + \left[ \nabla^2 \varphi(x)\right]^{-1} \nabla^2 f(x) \left[ \nabla^2 \varphi(x)\right]^{-1}.
  \end{equation*}
  We replace the third order term using the fourth point of proposition \ref{prop:legendre-properties} which gives, for any direction $r \in \mathbb{R}^n$:
  \begin{equation*}
    \nabla^3 \varphi^*(\nabla \varphi(x)) \left[r\right] = - \left[ \nabla^2 \varphi(x) \right]^{-1} \nabla^3 \varphi(x) \left[ \left[ \nabla^2 \varphi(x) \right]^{-1} r \right] \left[ \nabla^2 \varphi(x) \right]^{-1}
  \end{equation*}
  hence:
  \begin{equation*}
    \nabla^3 \varphi^* (\nabla \varphi(x)) \left[ \nabla f(x)\right] =  - \left[ \nabla^2 \varphi(x) \right]^{-1} \nabla^3 \varphi(x) \left[ \left[ \nabla^2 \varphi(x) \right]^{-1} \nabla f(x) \right] \left[ \nabla^2 \varphi(x) \right]^{-1}
  \end{equation*}
  thus:
  \begin{equation*}
    \nabla^2 (f \circ \nabla \varphi^*) (\eta) = \left[ \nabla^2 \varphi(x) \right]^{-1} \left( \nabla^2 f(x) - \nabla^3 \varphi(x) \left[ \left[ \nabla^2 \varphi(x) \right]^{-1} \nabla f(x) \right] \right) \left[ \nabla^2 \varphi(x) \right]^{-1}.
  \end{equation*}
  And finally:
  \begin{equation*}
    \left[ \nabla^2 \varphi(x) \right]^{-1} \left( \nabla^2 f(x) - \nabla^3 \varphi(x) \left[ \left[ \nabla^2 \varphi(x) \right]^{-1} \nabla f(x) \right] \right) \left[ \nabla^2 \varphi(x) \right]^{-1} \preceq L \left[ \nabla^2 \varphi(x) \right]^{-1}
  \end{equation*}
  Multiplying to the left and right by $\nabla^2 \varphi(x)$ gives the expected result.
\end{proof}
We can now compare the characterization of relative smoothness for twice differentiable functions in definition \ref{def:relative-smoothness}:
\begin{equation*}
  \nabla^2 f(x) \preceq L \nabla^2 \varphi(x)
\end{equation*}
and relative smoothness in mirror coordinates:
\begin{equation*}
  \nabla^2 f(x) - \nabla^3 \varphi(x) \left[ \left[ \nabla^2 \varphi(x) \right]^{-1} \nabla f(x) \right] \preceq L \nabla^2 \varphi(x).
\end{equation*}
It is clear that a function that is not $L$-relatively smooth in the primal coordinates can become $L$-relatively smooth in mirror coordinates when the term $\nabla^3 \varphi(x) \left[ \left[ \nabla^2 \varphi(x) \right]^{-1} \nabla f(x) \right]$ is positive.
Conversely, a function that is $L$-relatively smooth in the primal coordinates can lose this relative smoothness once transported in mirror space if the correction term is negative.
Neither condition seems to imply the other in general.
The mirror space relative smoothness is also reminiscent of self-concordant functions \cite[Chapter 5]{nesterovLecturesConvexOptimization2018}, but we have not investigated this link further.

\subsection{Application to Poisson inverse problems}
\label{subsec:poisson-application}

We now specialize \glsxtrshort{rbpg} to the Poisson inverse problem setting, where $f$ is the Poisson \glsxtrshort{nll} defined in equation \eqref{eq:poisson-nll}.
As potential for the Bregman divergence, we choose the weighted negative entropy given by:
\begin{equation}
  \varphi(x) = \sum_{i=1}^n s_i \left[ x_i \log x_i - x_i \right], \quad \quad \varphi^*(\eta) = \sum_{i=1}^n s_i \exp(\eta_i / s_i),
  \label{eq:weighted-negative-entropy}
\end{equation}
where $s = A^\top \mathbf{1}_m$ is the vector of column sums of $A$.
The Hessian of this potential is:
\begin{equation}
  \label{eq:weighted-negentropy-hessian}
  \nabla^2 \varphi(x) = \operatorname{diag}\left( \frac{s}{x}\right),
\end{equation}
The associated Bregman divergence is the weighted Kullback--Leibler divergence
\begin{equation}
  \label{eq:weighted-kl-divergence}
  D_\varphi(x \parallel y) = \sum_{i=1}^n s_i \left[ x_i \log \frac{x_i}{y_i} - x_i + y_i \right].
\end{equation}
The cornerstone of the approach proposed here is the mirror space relative smoothness between $f$ and $\varphi$, which we establish in the following lemma for the Poisson \glsxtrshort{nll} and the potential defined in \eqref{eq:weighted-negative-entropy}:

\begin{lemma}[Relative smoothness of the reparameterized Poisson objective]
  \label{lem:relative-smoothness-reparameterized}
  The Poisson \glsxtrshort{nll} reparameterized with respect to the weighted negative entropy \eqref{eq:weighted-negative-entropy},  $f \circ \nabla \varphi^*$,  is $1$-smooth relative to $\varphi^*$.
\end{lemma}
\begin{proof}
  The proof is deferred to appendix \ref{proof:lem:relative-smoothness-reparameterized}.
\end{proof}
The 1-relative smoothness of the reparameterized Poisson objective allows for a descent step-size $\tau \leq 1$.
In the unregularized case $g=0$, the preconditioned gradient step \eqref{eq:preconditioned-gradient-descent} becomes, by replacing the inverse of the Hessian of the potential by its expression \eqref{eq:weighted-negentropy-hessian}:
\begin{equation}
  x^{(k+1)} = x^{(k)} - \tau \frac{x^{(k)}}{s} \odot \nabla f(x^{(k)}) = x^{(k)} - \tau \frac{x^{(k)}}{s} \odot \left( A^\top \mathbf{1}_m - A^\top \frac{y}{Ax^{(k)}} \right)
\end{equation}
which can be rewritten as:
\begin{equation}
  \label{eq:multiplicative-updates-step-size}
  x^{(k+1)} = (1 - \tau) x^{(k)} + \tau \frac{x^{(k)}}{s} \odot A^\top \frac{y}{Ax^{(k)}}
\end{equation}
and taking the step size $\tau = 1$ recovers the well-known multiplicative updates \eqref{eq:multiplicative-updates}.
Now considering a convex regularization term $g$, and taking $\tau = 1$, the regularized counterpart of the multiplicative updates is
\begin{equation}
  x^{(k+1)} = \rightprox^{\varphi}_{g} \left( \frac{x^{(k)}}{s} \odot A^\top \frac{y}{Ax^{(k)}} \right).
  \label{eq:regularized-multiplicative-updates}
\end{equation}

These iterates can also be interpreted as a form of implicit multiplicative updates.
Let $z^{(k)} = \frac{x^{(k)}}{s} \odot A^\top \frac{y}{Ax^{(k)}}$ be the point at which we compute the right Bregman proximal map.
For strictly positive iterates, $x^{(k+1)}$ is the solution of the following optimization problem:
\begin{equation}
    \label{eq:right-bregman-prox-2}
    x^{(k+1)} = \argmin_{x \in \mathbb{R}^n_{++}} \left\{ g(x) + D_\varphi \left( z^{(k)} \parallel x \right) \right\}.
\end{equation}
Substituting the expression of the weighted \glsxtrshort{kl} divergence \eqref{eq:weighted-kl-divergence} into \eqref{eq:right-bregman-prox-2} gives:
\begin{equation*}
    x^{(k+1)} = \argmin_{x \in \mathbb{R}^n_{++}} \left\{ g(x) + \sum_{i=1}^n s_i \left[ z^{(k)}_i \log \frac{z^{(k)}_i}{x_i} + x_i - z^{(k)}_i \right] \right\},
\end{equation*}
which is equivalent, after discarding terms independent of $x$, to:
\begin{equation*}
    x^{(k+1)} = \argmin_{x \in \mathbb{R}^n_{++}} \left\{ g(x) + \sum_{i=1}^n s_i \left[ x_i - z^{(k)}_i \log (x_i) \right] \right\}.
\end{equation*}
The first-order optimality condition is, for some $\xi^{(k+1)} \in \partial g(x^{(k+1)})$:
\begin{equation*}
    0 = \xi^{(k+1)} + s \odot \left( \mathbf{1}_n - \frac{z^{(k)}}{x^{(k+1)}} \right)
\end{equation*}
and substituting the definition of $z^{(k)}$ gives:
\begin{equation*}
    0 = \xi^{(k+1)} + s - \frac{x^{(k)}}{x^{(k+1)}}A^\top \frac{y}{Ax^{(k)}}.
\end{equation*}
Equivalently,
\begin{equation}
    x^{(k+1)} \odot (s + \xi^{(k+1)})= x^{(k)} \odot A^\top \frac{y}{Ax^{(k)}}.
    \label{eq:implicit-multiplicative-gradient}
\end{equation}
When $z^{(k)}\in\mathbb{R}^n_{++}$, the positivity of $x^{(k+1)}$ and \eqref{eq:implicit-multiplicative-gradient} imply that $s+\xi^{(k+1)}\in\mathbb{R}^n_{++}$.
Thus, we can alternatively write:
\begin{equation}
  \label{eq:implicit-multiplicative}
    x^{(k+1)} = \frac{x^{(k)}}{s + \xi^{(k+1)}} \odot A^\top \frac{y}{Ax^{(k)}}.
\end{equation}

The implicit formulation \eqref{eq:implicit-multiplicative} can be related to Green's method \cite{greenUseEmAlgorithm1990}, a common heuristic for regularizing multiplicative updates, which is also sometimes called \glsxtrfull{osl}.
The name of this heuristic comes from the fact that it replaces $\xi^{(k+1)}\in\partial g(x^{(k+1)})$ by a subgradient $\xi^{(k)}\in\partial g(x^{(k)})$ evaluated at the current iterate, leading, whenever the denominator is positive, to
\begin{equation}
  x^{(k+1)} = \frac{x^{(k)}}{s+\xi^{(k)}} \odot A^\top\frac{y}{Ax^{(k)}}. \label{eq:osl-multiplicative-updates}
\end{equation}
This approximation gives inexpensive explicit updates and has been widely used in emission tomography \cite{greenBayesianReconstructionsEmission1990} and microscopy deconvolution \cite{deyRichardsonLucyAlgorithm2006}, but it does not generally guarantee convergence or descent.
We will use this implicit formulation again in the convergence analysis of section \ref{subsec:convergence-analysis-poisson}.

\section{Convergence analysis}
\label{sec:convergence-analysis-convex}

In this section, we study the convergence properties of the updates \eqref{eq:bpg-right}.
Under appropriate assumptions, we prove that the objective decreases along the iterates in the general convex setting.
We then specialize the analysis to Poisson inverse problems and to the weighted negative entropy potential \eqref{eq:weighted-negative-entropy}.
This yields sublinear convergence rates in function values for the regularized multiplicative updates \eqref{eq:regularized-multiplicative-updates}.

\subsection{Problem setting and assumptions}
\label{sec:analysis:subsec:assumptions}

We recall that we aim to minimize the composite convex objective $\Phi=f+g$ using the \glsxtrshort{rbpg} updates \eqref{eq:bpg-right}.
We impose the following standing assumptions.
\begin{assumption}[Baseline assumptions]
  \label{assumption:blanket}
  \leavevmode
  \begin{itemize}
    \item $f: \mathbb{R}^n \to \mathbb{R}\cup\{+\infty\}$ is proper, \glsxtrshort{lsc} and convex, and is differentiable on $\intdom\varphi$.
    \item $\varphi: \mathbb{R}^n \to \mathbb{R}\cup\{+\infty\}$ is of Legendre type and $\mathcal{C}^2$ on $\intdom\varphi$, with $\nabla^2\varphi(x) \succ 0$ for every $x \in \intdom\varphi$.
    \item $g: \mathbb{R}^n \to \mathbb{R}\cup\{+\infty\}$ is a proper, \glsxtrshort{lsc} and convex function with $\operatorname{dom} \varphi \subset \operatorname{dom} g$.
    \item $\inf \left\{ \Phi(x) : x \in \overline{\operatorname{dom}} \varphi \right\} > -\infty$.
  \end{itemize}
\end{assumption}

\begin{assumption}[Dual relative smoothness]
  \label{assumption:dual-relative-smoothness}
  The reparameterized function $f\circ\nabla\varphi^*: \mathbb{R}^n \to \mathbb{R} \cup \left\{+\infty \right\}$ is $L$-smooth relative to $\varphi^*$ for some $L>0$.
\end{assumption}

We now give sufficient conditions for the update to exist in $\intdom\varphi$ and be unique.
\begin{assumption}
  \label{assumption:well-definedness}
  Let $L$ be the mirror space relative smoothness constant.
  There is $\tau \in (0, 1/L]$ such that the following properties hold:
  \begin{itemize}
    \item For every $x\in\intdom\varphi$,
    \begin{equation*}
      x-\tau\left[\nabla^2\varphi(x)\right]^{-1}\nabla f(x)
      \in\intdom\varphi.
    \end{equation*}
    \item For every $z\in\intdom\varphi$, the extended-valued mapping
    \begin{equation*}
      u\mapsto
      \begin{cases}
        \tau g(u)+D_\varphi(z\parallel u), & u\in\intdom\varphi,\\
        +\infty, & u\notin\intdom\varphi
      \end{cases}
    \end{equation*}
    is \glsxtrshort{lsc} and its sublevel sets are bounded on $\mathbb{R}^n$.
    \item For every $z\in\intdom\varphi$, the right Bregman proximal mapping
    \begin{equation*}
      \rightprox^\varphi_{\tau g}(z)
      = \argmin_{u\in\intdom\varphi}
      \left\{\tau g(u)+D_\varphi(z\parallel u)\right\}
    \end{equation*}
    is single-valued.
  \end{itemize}
\end{assumption}
Let us define the mapping:
\begin{equation}
  T_\tau(x) = \argmin_{u \in \intdom\varphi} \left\{ g(u) + \left\langle \left[ \nabla^2 \varphi(x) \right]^{-1} \nabla f(x), \nabla \varphi(u) - \nabla \varphi(x) \right\rangle + \frac{1}{\tau} D_\varphi(x \parallel u) \right\}.
  \label{eq:general-mapping}
\end{equation}
\begin{lemma}[Well-definedness]
  Under assumptions \ref{assumption:blanket} and \ref{assumption:well-definedness}, $T_\tau(x)$ is well-defined and single-valued from $\intdom \varphi$ to $\intdom\varphi$.
  Moreover,
  \begin{equation}
    T_\tau(x)
    =\rightprox^\varphi_{\tau g}\!\left(
      x-\tau\left[\nabla^2\varphi(x)\right]^{-1}\nabla f(x)
    \right).
  \end{equation}
  \label{lem:well-definedness}
\end{lemma}
\begin{proof}
  Multiplying the objective in \eqref{eq:general-mapping} by $\tau$ and expanding $D_\varphi(x\parallel u)$, the function to minimize is equal, up to additive terms independent of $u$, to the extended-valued mapping in assumption \ref{assumption:well-definedness}, evaluated at
  \begin{equation*}
    z=x-\tau\left[\nabla^2\varphi(x)\right]^{-1}\nabla f(x).
  \end{equation*}
  This mapping is proper because $z\in\intdom\varphi$ and $\operatorname{dom}\varphi\subset\operatorname{dom}g$.
  By assumption \ref{assumption:well-definedness}, the function being minimized is \glsxtrshort{lsc}, with bounded, thus compact, sublevel sets.
  Hence, it attains its minimum on $\mathbb{R}^n$.
  Since the extended objective is infinite outside $\intdom\varphi$, every minimizer belongs to $\intdom\varphi$.
  Finally, the uniqueness of the minimizer results from assumption \ref{assumption:well-definedness}.
\end{proof}

Finding a minimal, universally relevant set of well-definedness conditions is not straightforward, even for the standard left Bregman proximal mapping.
See, for instance, the discussion around \cite[lemma~2]{bauschkeDescentLemmaLipschitz2017}.
Assumption \ref{assumption:well-definedness} is deliberately stated as a convenient sufficient condition for us to perform the analysis in the Poisson setting we are interested in.
Ideally, we would like to have a more general set of assumptions to guarantee that the natural-gradient step \eqref{eq:preconditioned-gradient-descent} remains in the interior of the domain of $\varphi$.
We verify that all of the points of assumption \ref{assumption:well-definedness} hold in the Poisson setting in proposition \ref{prop:poisson-assumptions}.

\subsection{Properties in the general convex case}

Since our algorithm can be viewed as a \glsxtrshort{bpg} update applied to a reparameterized objective, a natural starting point would be to apply existing convergence results for convex \glsxtrshort{bpg} \cite{bauschkeDescentLemmaLipschitz2017,luRelativelySmoothConvex2018,teboulleSimplifiedViewFirst2018}.
This would yield sublinear convergence rates in function values \cite[theorem 4.1]{teboulleSimplifiedViewFirst2018}, and convergence of the sequence of mirror iterates $\left( \eta^{(k)} \right)_{k \in \mathbb{N}}$ to a minimizer of the reparameterized objective $\Phi \circ \nabla \varphi^*$ \cite[theorem 4.2]{teboulleSimplifiedViewFirst2018}.
One could then relate minimizers of $\Phi \circ \nabla \varphi^*$ to those of the original objective $\Phi$ and translate these guarantees back to the primal iterates $\left( x^{(k)} \right)_{k \in \mathbb{N}}$.

However, a major obstacle to this approach is that the reparameterized objective $\Phi \circ \nabla \varphi^*$ need not be convex, even when $\Phi$ is.
Indeed, requiring the change of variables $\nabla \varphi^*$ to preserve convexity is a strong assumption, and it fails, for instance, for the weighted negative entropy \eqref{eq:weighted-negative-entropy} we use in the Poisson setting.
A possible alternative is to rely on the convergence analysis for non-convex \glsxtrshort{bpg} \cite{bolteFirstOrderMethods2018}.
This route, however, requires additional assumptions on the regularization term $g$ which are too restrictive to be practical.
For completeness, we detail this alternate analysis of \glsxtrshort{rbpg} using tools from the non-convex literature in Appendix \ref{appendix:sec:non-convex}.

In the rest of this section, we therefore follow yet a different route.
Rather than working in the reparameterized mirror coordinates, we work directly in primal coordinates and exploit the convexity of $\Phi$ to establish descent properties for the iterates $\left( x^{(k)} \right)_{k \in \mathbb{N}}$.
This approach uses only the standing assumptions \ref{assumption:blanket}--\ref{assumption:well-definedness}.

We first recall a well-known three-point identity for Bregman divergences:
\begin{lemma}[Three-point identity {\cite[lemma 3.1]{chenConvergenceAnalysisProximalLike1993}}]
    \label{lem:three-point-identity}
    Let $\varphi$ be a Legendre function and $D_\varphi$ its associated Bregman divergence.
    Then, for all $x, y, z \in \intdom \varphi$:
    \begin{equation}
        D_\varphi(x \parallel z) = D_\varphi(x \parallel y) + D_\varphi(y \parallel z) + \langle \nabla \varphi(y) - \nabla \varphi(z), x - y \rangle.
    \end{equation}
    Alternatively, taking $z = x$ gives a two-point identity:
    \begin{equation}
        D_\varphi(x \parallel y) + D_\varphi(y \parallel x) = \langle \nabla \varphi(y) - \nabla \varphi(x), y - x \rangle.
    \end{equation}
\end{lemma}
\begin{proof}
    Expanding the definition of the Bregman divergences gives the expected result, see \cite[lemma 3.1]{chenConvergenceAnalysisProximalLike1993} for detailed computations.
\end{proof}
We will also need the following descent inequality in primal coordinates, which is an immediate consequence of the relative smoothness of the reparameterized function $f \circ \nabla \varphi^*$:
\begin{lemma}[Descent inequality]
  \label{lem:descent-inequality-reparameterized}
  Suppose assumptions \ref{assumption:blanket} and \ref{assumption:dual-relative-smoothness} hold.
  Then for all $u,x\in\intdom\varphi$, we have:
  \begin{equation}
    f(u) \leq f(x) + \langle \left[ \nabla^2 \varphi(x) \right]^{-1} \nabla f(x), \nabla \varphi(u) - \nabla \varphi(x) \rangle + L D_\varphi(x \parallel u).
  \end{equation}
\end{lemma}
\begin{proof}
  By applying the descent inequality \eqref{lem:descent} to the reparameterized objective $f \circ \nabla \varphi^*$, we have for any $\eta,\, \xi\in\intdom\varphi^*$:
  \begin{equation*}
    (f \circ \nabla \varphi^*)(\xi) \leq (f \circ \nabla \varphi^*)(\eta) + \langle \nabla (f \circ \nabla \varphi^*)(\eta), \xi - \eta \rangle + L D_{\varphi^*}(\xi \parallel \eta).
  \end{equation*}
  We use the change of variables $u=\nabla\varphi^*(\xi)$ and $x = \nabla \varphi^*(\eta)$ to rewrite the inequality as:
  \begin{equation*}
    f(u) \leq f(x) + \langle \nabla (f \circ \nabla \varphi^*)(\eta), \nabla \varphi(u) - \nabla \varphi(x) \rangle + L D_{\varphi^*}(\nabla \varphi(u) \parallel \nabla \varphi(x)).
  \end{equation*}
  The desired inequality follows from the second point of proposition \ref{prop:legendre-properties} and equation \eqref{eq:reparameterized-gradient} which gives:
  \begin{equation*}
    D_{\varphi^*}(\nabla \varphi(u) \parallel \nabla \varphi(x)) = D_\varphi(x \parallel u).
  \end{equation*}
\end{proof}
With these two lemmas, we are now ready to prove the decrease of the composite objective $\Phi$ along the iterates.
\begin{proposition}[Function value decrease along the iterates]
  \label{prop:objective-decrease}
  Suppose assumptions \ref{assumption:blanket} and \ref{assumption:dual-relative-smoothness} hold, and let $0<\tau\leq 1/L$ satisfy assumption \ref{assumption:well-definedness}.
  Starting from $x^{(0)}\in\intdom\varphi$, let $\left\{ x^{(k)} \right\}_{k \in \mathbb{N}}$ be generated by $x^{(k+1)}=T_\tau(x^{(k)})$.
  Then the sequence of function values
  $\left( \Phi(x^{(k)}) \right)_{k \in \mathbb{N}}$ is non-increasing and satisfies
  \begin{equation}
    \Phi(x^{(k+1)}) - \Phi(x^{(k)}) \le \left(L-\frac{1}{\tau}\right) D_\varphi(x^{(k)} \parallel x^{(k+1)}).
  \end{equation}
\end{proposition}

\begin{proof}
  Let us define:
  \begin{equation*}
    z^{(k)}=x^{(k)}-\tau\left[\nabla^2\varphi(x^{(k)})\right]^{-1}
    \nabla f(x^{(k)}).
  \end{equation*}
  Applying lemma \ref{lem:descent-inequality-reparameterized} gives
  \begin{align*}
    f(x^{(k+1)})-f(x^{(k)}) \leq &\left\langle \left[\nabla^2\varphi(x^{(k)})\right]^{-1}\nabla f(x^{(k)}), \nabla\varphi(x^{(k+1)})-\nabla\varphi(x^{(k)}) \right\rangle \\
                                &+L D_\varphi(x^{(k)}\parallel x^{(k+1)}),
  \end{align*}
  which can be rewritten as:
  \begin{equation}
    \label{eq:descent-primal}
    f(x^{(k+1)})-f(x^{(k)}) \leq \frac{1}{\tau} \left\langle z^{(k)}-x^{(k)}, \nabla\varphi(x^{(k)}) - \nabla\varphi(x^{(k+1)})\right\rangle + L D_\varphi(x^{(k)}\parallel x^{(k+1)}).
  \end{equation}
  Applying the three-point identity of lemma \ref{lem:three-point-identity} with $x=z^{(k)}$, $y=x^{(k)}$, and $z=x^{(k+1)}$ gives:
  \begin{align*}
    \left\langle z^{(k)}-x^{(k)}, \nabla \varphi(x^{(k)}) - \nabla \varphi(x^{(k+1)}) \right\rangle = ~&D_\varphi(z^{(k)}\parallel x^{(k+1)}) \\
    &-D_\varphi(z^{(k)}\parallel x^{(k)}) -D_\varphi(x^{(k)}\parallel x^{(k+1)}),
  \end{align*}
  which combined with \eqref{eq:descent-primal} yields:
  \begin{equation}
    f(x^{(k+1)})-f(x^{(k)}) \leq \frac{1}{\tau} \left(D_\varphi(z^{(k)}\parallel x^{(k+1)}) -D_\varphi(z^{(k)}\parallel x^{(k)}) \right) + \left( L - \frac{1}{\tau} \right) D_\varphi(x^{(k)}\parallel x^{(k+1)}).
    \label{eq:proof:function-value-decrease:eq-f}
  \end{equation}
  We will now upper bound the $g$ part of the objective.
  Recall that:
  \begin{equation*}
    x^{(k+1)}=\rightprox^\varphi_{\tau g}(z^{(k)}),
  \end{equation*}
  hence by definition of the minimum we have:
  \begin{equation}
    g(x^{(k+1)}) - g(x^{(k)}) \leq \frac{1}{\tau} \left(D_\varphi(z^{(k)}\parallel x^{(k)}) - D_\varphi(z^{(k)}\parallel x^{(k+1)})\right).
    \label{eq:proof:function-value-decrease:eq-g}
  \end{equation}
  We complete the proof by summing both inequalities \eqref{eq:proof:function-value-decrease:eq-f} and \eqref{eq:proof:function-value-decrease:eq-g} and obtain:
  \begin{equation}
    \Phi(x^{(k+1)}) - \Phi(x^{(k)}) \leq \left(L-\frac{1}{\tau}\right) D_\varphi(x^{(k)}\parallel x^{(k+1)}).
  \end{equation}
\end{proof}
Under these general assumptions, the preceding argument establishes monotonicity but does not yield a function-value convergence rate.

\subsection{Convergence rate in the Poisson setting}
\label{subsec:convergence-analysis-poisson}

We now focus on the Poisson inverse problem setting, where $f$ is the Poisson \glsxtrshort{nll} defined in \eqref{eq:poisson-nll}, and $\varphi$ is the weighted negative entropy defined in \eqref{eq:weighted-negative-entropy}.
The descent inequality derived in the general convex setting will remain a key ingredient in what follows.
To obtain a convergence rate in function values for the Poisson case, we combine it with the recent analysis developed in \cite{zhaoConvergenceRateAnalysis2023} for the multiplicative updates.
We follow the same proof pattern and extend it to the regularized setting and the updates \eqref{eq:regularized-multiplicative-updates}.

Throughout the following analysis, we recall that we assume $A$ to be a nonnegative matrix with at least one positive entry in each row and each column, and that the data $y \in \mathbb{R}^m_{++}$.
We now verify that the assumptions in section \ref{sec:analysis:subsec:assumptions} hold for this setting.
\begin{proposition}
    \label{prop:poisson-assumptions}
    Let $g:\mathbb{R}^n\to\mathbb{R}\cup\{+\infty\}$ be proper, \glsxtrshort{lsc}, convex, and bounded from below, with $\operatorname{dom}\varphi\subset\operatorname{dom}g$, such that $\Phi=f+g$ has bounded sublevel sets on $\mathbb{R}^n_+$.
    Then assumptions \ref{assumption:blanket}, \ref{assumption:dual-relative-smoothness}, and \ref{assumption:well-definedness} hold with $L=1$ and $\tau=1$.
    Moreover, the solution set $\arg\min_{x\in\mathbb{R}^n_+}\Phi(x)$ is nonempty.
\end{proposition}
\begin{proof}
    The Poisson \glsxtrshort{nll} $f$ is proper, \glsxtrshort{lsc}, convex, differentiable on $\mathbb{R}^n_{++}$, and bounded from below.
    Since every column of $A$ contains a positive entry, $s=A^\top\mathbf{1}_m\in\mathbb{R}^n_{++}$.
    Hence the weighted negative entropy $\varphi$ is of Legendre type, is $\mathcal{C}^2$ on $\mathbb{R}^n_{++}$, and satisfies
    \begin{equation*}
        \nabla^2\varphi(x)=\operatorname{diag}(s/x)\succ0.
    \end{equation*}
    Since both $f$ and $g$ are bounded from below, $\inf_{x\in\mathbb{R}^n_+}\Phi(x)>-\infty$.
    Together with the assumptions on $g$, this verifies assumption \ref{assumption:blanket}.
    Assumption \ref{assumption:dual-relative-smoothness} holds with $L=1$ by lemma \ref{lem:relative-smoothness-reparameterized}.
    It remains to verify assumption \ref{assumption:well-definedness} for $\tau=1$.
    For every $x\in\mathbb{R}^n_{++}$, the preconditioned-gradient point is
    \begin{equation*}
        x-\left[\nabla^2\varphi(x)\right]^{-1}\nabla f(x)
        =\frac{x}{s}\odot A^\top\frac{y}{Ax}\in\mathbb{R}^n_{++},
    \end{equation*}
    where strict positivity follows from the assumptions on $A$ and $y$.
    For every $z\in\mathbb{R}^n_{++}$, the mapping
    \begin{equation}
      \label{eq:proximal-weighted-negentropy-mapping}
        u\mapsto g(u)+D_\varphi(z\parallel u)
    \end{equation}
    extended by $+\infty$ outside $\mathbb{R}^n_{++}$ is \glsxtrshort{lsc} and its sublevel sets are bounded, because $g$ is \glsxtrshort{lsc} and bounded from below while $D_\varphi(z\parallel\cdot)$ diverges at the boundary of the positive orthant and as $\|u\|\to+\infty$.

    More precisely, let $z\in\mathbb{R}^n_{++}$ and $u\in\mathbb{R}^n_+\setminus\mathbb{R}^n_{++}$.
    There exists an index $i$ such that $u_i=0$, while $s_i z_i>0$.
    Using the conventions for the generalized \glsxtrshort{kl} divergence stated in the introduction, the corresponding term $s_i z_i\log(z_i/u_i)$ is equal to $+\infty$, and hence $D_\varphi(z\parallel u)=+\infty$.
    Hence, none of the minimizers of \eqref{eq:proximal-weighted-negentropy-mapping} can lie on the boundary of the positive orthant.

    Moreover, $D_\varphi(z\parallel\cdot)$ is strictly convex on $\mathbb{R}^n_{++}$.
    Since $g$ is convex, the proximal objective is therefore strictly convex, so its minimizer is unique and the right Bregman proximal mapping is single-valued.
    Thus assumption \ref{assumption:well-definedness} holds.

    Finally, $\Phi$ is proper and \glsxtrshort{lsc}, and its sublevel sets are bounded by assumption.
    Thus, the solution set $\arg\min_{x\in\mathbb{R}^n_+}\Phi(x)$ is nonempty.
\end{proof}

Proposition \ref{prop:poisson-assumptions} allows us to apply the results of the previous subsection.
In particular, starting from $x^{(0)}\in\mathbb{R}^n_{++}$, the natural-gradient point \eqref{eq:preconditioned-gradient-descent}, given by $z^{(k)} = \frac{x^{(k)}}{s} \odot A^\top \frac{y}{Ax^{(k)}}$ in the Poisson case, is strictly positive at every iteration, and the subsequent right Bregman proximal step cannot return a point on the boundary of the positive orthant.
Therefore the regularized multiplicative updates are well-defined, remain in $\mathbb{R}^n_{++}$, and generate a non-increasing sequence of objective values.
For the following analysis, we rely on the implicit formulation \eqref{eq:implicit-multiplicative} derived in section \ref{subsec:poisson-application}, which we recall here:
\begin{equation*}
    x^{(k+1)} = \frac{x^{(k)}}{s + \xi^{(k+1)}} \odot A^\top \frac{y}{Ax^{(k)}},
    \qquad \xi^{(k+1)}\in\partial g(x^{(k+1)}).
\end{equation*}
Equivalently, the iterates satisfy the optimality relation \eqref{eq:implicit-multiplicative-gradient}:
\begin{equation*}
    x^{(k+1)} \odot (s + \xi^{(k+1)}) = x^{(k)} \odot A^\top \frac{y}{Ax^{(k)}}.
\end{equation*}

For every $a\in\mathbb{R}^n_{++}$, we define the associated scaled simplex
\begin{equation*}
    \Delta_n(a):=\left\{x\in\mathbb{R}^n_+:\langle a,x\rangle=\|y\|_1\right\}.
\end{equation*}
In the unregularized case where $g=0$, it is well known that the multiplicative updates \eqref{eq:multiplicative-updates} stay within the scaled simplex $\Delta_n(s)$.
Indeed, a direct computation gives:
\begin{align*}
  \langle s, x^{(k+1)} \rangle &= \langle s, \frac{x^{(k)}}{s} \odot A^\top \frac{y}{Ax^{(k)}} \rangle \\
  &= \langle x^{(k)}, A^\top \frac{y}{A x^{(k)}} \rangle \\
  &= \langle A x^{(k)}, \frac{y}{A x^{(k)}} \rangle \\
  &= \| y \|_1.
\end{align*}
A similar property holds in the regularized case and directly follows from the implicit gradient expression \eqref{eq:implicit-multiplicative-gradient} of the updates:
\begin{lemma}[Invariant]
    \label{lem:invariant}
    Let $x^{(k)}$ and $x^{(k+1)}$ be two consecutive iterates of the regularized multiplicative updates \eqref{eq:regularized-multiplicative-updates}, and let $\xi^{(k+1)} \in \partial g(x^{(k+1)})$ be such that
    \begin{equation}
        x^{(k+1)} \odot (s + \xi^{(k+1)})= x^{(k)} \odot A^\top \frac{y}{Ax^{(k)}}.
    \end{equation}
    Then $s+\xi^{(k+1)}\in\mathbb{R}^n_{++}$ and $x^{(k+1)}\in\Delta_n(s+\xi^{(k+1)})$.
\end{lemma}
\begin{proof}
    \begin{align*}
        \langle s + \xi^{(k+1)}, x^{(k+1)} \rangle &= \langle (s + \xi^{(k+1)}) \odot x^{(k+1)}, \mathbf{1}_n \rangle \\
        &= \langle x^{(k)} \odot A^\top \frac{y}{Ax^{(k)}}, \mathbf{1}_n \rangle \\
        &= \langle A^\top \frac{y}{Ax^{(k)}}, x^{(k)} \rangle \\
        &= \langle \frac{y}{Ax^{(k)}}, A x^{(k)} \rangle \\
        &= \| y \|_1.
    \end{align*}
    The defining relation of $\xi^{(k+1)}$ in \eqref{eq:implicit-multiplicative-gradient}, together with the strict positivity of the iterates and of $A^\top(y/Ax^{(k)})$, also gives $s+\xi^{(k+1)}\in\mathbb{R}^n_{++}$.
\end{proof}
The second technical lemma is the analogue of \cite[lemma 1]{zhaoConvergenceRateAnalysis2023} and is obtained by writing the \glsxtrfull{kkt} conditions of the composite optimization problem \eqref{eq:composite-optimization-convex} for the Poisson \glsxtrshort{nll}.
\begin{lemma}[\glsxtrshort{kkt} condition]
    \label{lem:kkt-condition}
    For every solution $x^\star$ of the composite problem \eqref{eq:composite-optimization-convex} there exist $\xi^\star \in \partial g(x^\star)$ and $\nu^\star \in \mathbb{R}_+^n$ satisfying the \glsxtrshort{kkt} conditions
    \begin{equation}
        \nabla f(x^\star) + \xi^\star = \nu^\star
        \qquad\text{and}\qquad
        \nu_i^\star x_i^\star = 0
        \quad\text{for every }i\in\{1,\ldots,n\}.
    \end{equation}
    In particular, if $x^\star_i>0$ for some $i\in\{1,\ldots,n\}$, then $\nu_i^\star=0$ and $\nabla_i f(x^\star)=-\xi_i^\star$.
    Moreover, $s+\xi^\star\in\mathbb{R}^n_{++}$ and $x^\star\in\Delta_n(s+\xi^\star)$.
\end{lemma}
\begin{proof}
    The minimizer $x^\star \in \mathbb{R}^n_{+}$ has finite objective value $f(x^\star) < + \infty$. Since $y \in \mathbb{R}^m_{++}$, it follows that $Ax^\star \in \mathbb{R}^m_{++}$.
    Under the positivity constraint $x\in\mathbb{R}^n_+$, the Lagrangian of the composite optimization problem \eqref{eq:composite-optimization-convex} is
    \begin{equation*}
        \mathcal{L}(x, \nu) = f(x) + g(x) + \langle \nu, -x \rangle,
    \end{equation*}
    where $\nu \in \mathbb{R}^n_{+}$ is the vector of Lagrange multipliers associated with the nonnegativity constraint.
    The \glsxtrshort{kkt} conditions for optimality are
    \begin{equation}
        0 \in \partial g(x^\star) + \nabla f(x^\star) - \nu^\star
        \qquad\text{and}\qquad
        \nu_i^\star x_i^\star = 0
        \quad\text{for every }i\in\{1,\ldots,n\}.
    \end{equation}
    Using $\nabla f(x^\star)=s-A^\top(y/Ax^\star)$, we obtain, for some $\xi^\star \in \partial g(x^\star)$:
    \begin{equation*}
        s+\xi^\star=A^\top\frac{y}{Ax^\star}+\nu^\star.
    \end{equation*}
    Since $y/(Ax^\star)\in\mathbb{R}^m_{++}$ and every column of $A$ contains a positive entry, each component of $A^\top(y/Ax^\star)$ is strictly positive.
    Together with $\nu^\star \in\mathbb{R}^n_+$, this yields $s+\xi^\star\in\mathbb{R}^n_{++}$.
    Taking the inner product with $x^\star$ and using the complementarity conditions gives
    \begin{equation*}
        \langle x^\star,s+\xi^\star\rangle
        =\left\langle Ax^\star,\frac{y}{Ax^\star}\right\rangle
        +\langle x^\star,\nu^\star\rangle
        =\|y\|_1.
    \end{equation*}
    Therefore, $x^\star\in\Delta_n(s+\xi^\star)$.
\end{proof}
Finally, the last technical lemma is the regularized analogue of \cite[lemma 3]{zhaoConvergenceRateAnalysis2023}:

\begin{lemma}[Objective gap bound]
    \label{lem:objective-gap-bound}
    Let $x^\star$ be a solution of the composite problem \eqref{eq:composite-optimization-convex}, and let $\xi^\star \in \partial g(x^\star)$ and $\nu^\star \in \mathbb R_+^n$ such that
    \begin{equation*}
        s + \xi^\star = A^\top \left( \frac{y}{Ax^\star} \right) + \nu^\star,
        \qquad
        \nu_i^\star x_i^\star = 0
        \quad\text{for every }i\in\{1,\ldots,n\}.
    \end{equation*}
    Let $x\in\mathbb R_{++}^n\cap\Delta_n(s+\xi)$ where $\xi \in \partial g(x)$ satisfies $(s+\xi) \in \mathbb R_{++}^n$.
    Then
    \begin{equation}
        \Phi(x) - \Phi(x^\star) \leq \sum_{i=1}^n x_i^\star (s_i+\xi_i^\star) \log \left( \frac{[A^\top(y/Ax)]_i}{s_i+\xi_i} \right).
    \end{equation}
\end{lemma}

\begin{proof}
    We start from
    \begin{equation}
      \label{eq:rewrite-poisson-nll}
      \begin{aligned}
        f(x)-f(x^\star)
        &=\sum_{i=1}^m\left[(Ax)_i-(Ax^\star)_i
        -y_i\log\left(\frac{(Ax)_i}{(Ax^\star)_i}\right)\right] \\
        &=\langle s,x-x^\star\rangle
        +\sum_{i=1}^m y_i\log\left(\frac{(Ax^\star)_i}{(Ax)_i}\right).
      \end{aligned}
    \end{equation}
    Since $\xi \in \partial g(x)$, convexity of $g$ gives
    \begin{equation*}
        g(x)-g(x^\star) \le \langle \xi, x-x^\star \rangle,
    \end{equation*}
    hence
    \begin{equation}
        \Phi(x)-\Phi(x^\star) \leq \langle s+\xi,x-x^\star\rangle
        +\sum_{i=1}^m y_i\log\left(\frac{(Ax^\star)_i}{(Ax)_i}\right).
        \label{eq:objective-gap-preliminary-1}
    \end{equation}
    Then note that, with $x\in\Delta_n(s+\xi)$ by assumption, and $x^\star\in\Delta_n(s+\xi^\star)$ by lemma \ref{lem:kkt-condition}, we can rewrite the first term as:
    \begin{equation}
        \langle s+\xi,x-x^\star\rangle
        =\langle x^\star,\xi^\star-\xi\rangle,
        \label{eq:scaled-simplex-linear-identity}
    \end{equation}
    Substituting this identity into \eqref{eq:objective-gap-preliminary-1} yields:
    \begin{equation}
        \Phi(x)-\Phi(x^\star) \leq \langle x^\star, \xi^\star - \xi \rangle + \sum_{i=1}^m y_i \log\left(\frac{(Ax^\star)_i}{(Ax)_i}\right).
        \label{eq:objective-gap-preliminary}
    \end{equation}
    We now bound the logarithmic term.
    Let
    \begin{equation*}
        J^\star:=\left\{j\in\{1,\ldots,n\}:x_j^\star>0\right\}.
    \end{equation*}
    For every $j\in J^\star$, the \glsxtrshort{kkt} conditions give
    \begin{equation*}
        s_j+\xi_j^\star = \left[A^\top\left(\frac{y}{Ax^\star}\right)\right]_j.
    \end{equation*}
    Define the nonnegative coefficients
    \begin{equation*}
        w_{ij} := \frac{y_i a_{ij}}{(s_j+\xi_j^\star)(Ax^\star)_i}
    \end{equation*}
    which, for every $j\in J^\star$, satisfy $\sum_{i=1}^m w_{ij}=1$ and
    \begin{equation*}
        \frac{[A^\top(y/Ax)]_j}{s_j+\xi_j^\star} = \sum_{i=1}^m w_{ij} \frac{(Ax^\star)_i}{(Ax)_i}.
    \end{equation*}
    By concavity of the logarithm, for every $j\in J^\star$ we have
    \begin{equation*}
      \begin{aligned}
        \log\left(\frac{[A^\top(y/Ax)]_j}{s_j+\xi_j^\star}\right)
        &=\log\left(\sum_{i=1}^m
        w_{ij}\frac{(Ax^\star)_i}{(Ax)_i}\right) \\
        &\geq\sum_{i=1}^m
        w_{ij}\log\left(\frac{(Ax^\star)_i}{(Ax)_i}\right).
      \end{aligned}
    \end{equation*}
    Multiplying by $x_j^\star(s_j+\xi_j^\star)$ and summing over $j\in J^\star$, while noting that $x_j^\star=0$ for $j\notin J^\star$, gives
    \begin{align*}
        &\sum_{j=1}^n x_j^\star(s_j+\xi_j^\star)
        \log\left(\frac{[A^\top(y/Ax)]_j}{s_j+\xi_j^\star}\right) \\
        &\quad\geq \sum_{j\in J^\star}x_j^\star(s_j+\xi_j^\star)
        \sum_{i=1}^m w_{ij}\log\left(\frac{(Ax^\star)_i}{(Ax)_i}\right) \\
        &\quad= \sum_{i=1}^m\left(\sum_{j\in J^\star}
        x_j^\star(s_j+\xi_j^\star)w_{ij}\right)
        \log\left(\frac{(Ax^\star)_i}{(Ax)_i}\right).
    \end{align*}
    Moreover,
    \begin{align*}
        \sum_{j\in J^\star}x_j^\star(s_j+\xi_j^\star)w_{ij}
        &=\sum_{j\in J^\star}x_j^\star
        \frac{y_i a_{ij}}{(Ax^\star)_i} \\
        &=y_i\frac{\sum_{j\in J^\star}a_{ij}x_j^\star}
        {(Ax^\star)_i} \\
        &=y_i.
    \end{align*}
    Using this identity, we conclude that
    \begin{equation*}
        \sum_{j=1}^n x_j^\star(s_j+\xi_j^\star)
        \log\left(\frac{[A^\top(y/Ax)]_j}{s_j+\xi_j^\star}\right) \geq
        \sum_{i=1}^m y_i
        \log\left(\frac{(Ax^\star)_i}{(Ax)_i}\right).
    \end{equation*}
    Combining this bound with \eqref{eq:objective-gap-preliminary} and \eqref{eq:scaled-simplex-linear-identity} yields
    \begin{equation*}
        \Phi(x)-\Phi(x^\star)
        \leq \langle x^\star,\xi^\star-\xi\rangle + \sum_{j=1}^n x_j^\star(s_j+\xi_j^\star)
        \log\left(\frac{[A^\top(y/Ax)]_j}{s_j+\xi_j^\star}\right).
    \end{equation*}
    We decompose the logarithmic term as
    \begin{align*}
        \sum_{j=1}^n x_j^\star(s_j+\xi_j^\star)
        \log\left(\frac{[A^\top(y/Ax)]_j}{s_j+\xi_j^\star}\right) =& ~
        \sum_{j=1}^n x_j^\star(s_j+\xi_j^\star)
        \log\left(\frac{[A^\top(y/Ax)]_j}{s_j+\xi_j}\right) \\
        &+
        \sum_{j=1}^n x_j^\star(s_j+\xi_j^\star)
        \log\left(\frac{s_j+\xi_j}{s_j+\xi_j^\star}\right).
    \end{align*}
    Moreover,
    \begin{align*}
        \langle x^\star,\xi^\star-\xi\rangle
        +\sum_{j=1}^n x_j^\star(s_j+\xi_j^\star)
        \log\left(\frac{s_j+\xi_j}{s_j+\xi_j^\star}\right)
        &= \sum_{j=1}^n x_j^\star\left[
        (s_j+\xi_j^\star)\log\left(\frac{s_j+\xi_j}{s_j+\xi_j^\star}\right)
        +\xi_j^\star-\xi_j\right] \\
        &=-\mathrm{KL}\left(
        x^\star\odot(s+\xi^\star)
        \parallel x^\star\odot(s+\xi)\right).
    \end{align*}
    Substituting the two preceding identities into the objective-gap bound \eqref{eq:objective-gap-preliminary} gives
    \begin{equation*}
        \Phi(x)-\Phi(x^\star)
        \leq \sum_{j=1}^n x_j^\star(s_j+\xi_j^\star)
        \log\left(\frac{[A^\top(y/Ax)]_j}{s_j+\xi_j}\right) -\mathrm{KL}\left(
        x^\star\odot(s+\xi^\star)
        \parallel x^\star\odot(s+\xi)\right),
    \end{equation*}
    and since the \glsxtrshort{kl} divergence is non-negative, we get the desired result.
\end{proof}

With these three technical lemmas, we are finally ready to state our main convergence result, which is the sublinear rate $\mathcal{O}(1/k)$ in function values of the regularized multiplicative updates \eqref{eq:regularized-multiplicative-updates}.
We follow the proof pattern of \cite{zhaoConvergenceRateAnalysis2023} and make use of the previous technical lemmas to obtain a fairly simple proof using a telescoping argument.
\begin{theorem}
    \label{theorem:function-value-rate}
    Suppose the hypotheses of proposition \ref{prop:poisson-assumptions} hold.
    Let $\left(x^{(k)}\right)_{k\in\mathbb{N}}$ be the sequence generated by the regularized multiplicative updates \eqref{eq:implicit-multiplicative} from an initial point $x^{(0)}\in\mathbb{R}^n_{++}$.
    Let $x^\star\in\arg\min_{x\in\mathbb{R}^n_+}\Phi(x)$, and let $\xi^\star\in\partial g(x^\star)$ be a \glsxtrshort{kkt} subgradient given by lemma \ref{lem:kkt-condition}.
    Then, for every $k\geq1$,
    \begin{equation*}
        \Phi(x^{(k)}) - \Phi(x^\star) \leq \frac{\mathrm{KL}( x^\star \odot (s + \xi^\star) \parallel x^{(1)} \odot (s + \xi^{(1)}))}{k}.
    \end{equation*}
    Moreover, the ergodic iterate $\overline{x}^{(k)} = \frac{1}{k} \sum_{\ell=1}^k x^{(\ell)}$ satisfies
    \begin{equation*}
        \Phi(\overline{x}^{(k)}) - \Phi(x^\star) \leq \frac{\mathrm{KL}( x^\star \odot (s + \xi^\star) \parallel x^{(1)} \odot (s + \xi^{(1)}))}{k}.
    \end{equation*}
\end{theorem}
\begin{proof}
    Let us define the Lyapunov energy:
    \begin{equation}
        \mathcal{E}_k := \mathrm{KL}( x^\star \odot (s + \xi^\star) \parallel x^{(k)} \odot (s + \xi^{(k)})).
    \end{equation}
    Since $x^{(k)}\in\Delta_n(s+\xi^{(k)})$ for all $k\geq1$ by lemma \ref{lem:invariant}, while $x^\star\in\Delta_n(s+\xi^\star)$ by lemma \ref{lem:kkt-condition}, the expression of this \glsxtrshort{kl} divergence simplifies to just the logarithmic term:
    \begin{equation}
        \mathcal{E}_k = \sum_{i=1}^n x_i^\star (s_i+\xi_i^\star) \log \left( \frac{x_i^\star (s_i + \xi_i^\star)}{x_i^{(k)} (s_i + \xi_i^{(k)})} \right).
    \end{equation}
    Now we study the difference $\mathcal{E}_k - \mathcal{E}_{k+1}$:
    \begin{equation}
      \label{eq:difference-1}
        \mathcal{E}_k - \mathcal{E}_{k+1} = \sum_{i=1}^n x_i^\star (s_i+\xi_i^\star) \log \left( \frac{x_i^{(k+1)} (s_i + \xi_i^{(k+1)})}{x_i^{(k)} (s_i + \xi_i^{(k)})} \right).
    \end{equation}
    Replacing $x^{(k+1)}$ with the implicit formulation \eqref{eq:implicit-multiplicative}, the difference \eqref{eq:difference-1} can be written as:
    \begin{equation}
        \mathcal{E}_k - \mathcal{E}_{k+1} = \sum_{i=1}^n x_i^\star (s_i+\xi_i^\star) \log \left( \frac{[A^\top(y/Ax^{(k)})]_i}{s_i + \xi_i^{(k)}} \right).
    \end{equation}
    Lemma \ref{lem:objective-gap-bound} gives:
    \begin{equation}
        0 \leq \Phi(x^{(k)}) - \Phi(x^\star) \leq \mathcal{E}_k - \mathcal{E}_{k+1}.
    \end{equation}
    We can now telescope over $k$ to get
    \begin{equation*}
        \sum_{l=1}^k \left[ \Phi(x^{(l)}) - \Phi(x^\star) \right] \leq \mathcal{E}_1 - \mathcal{E}_{k+1} \leq \mathcal{E}_1,
    \end{equation*}
    and since $\left( \Phi(x^{(k)}) \right)_{k \in \mathbb{N}}$ is non-increasing (proposition \ref{prop:objective-decrease}), we have:
    \begin{equation*}
        k \left( \Phi(x^{(k)}) - \Phi(x^\star) \right) \leq \sum_{l=1}^k \left[ \Phi(x^{(l)}) - \Phi(x^\star) \right] \leq \mathcal{E}_1,
    \end{equation*}
    which gives the first bound.
    Using the convexity of $\Phi$ we have
    \begin{equation*}
        k \Phi(\overline{x}^{(k)}) \leq \sum_{l=1}^k \Phi(x^{(l)}),
    \end{equation*}
    which further proves the second bound.
\end{proof}

\section{Numerical experiments}
\label{sec:experiments}

In this section we first derive the updates of our algorithm for several usual choices of regularizers $g$.
We then compare its performance against other algorithms for regularized Poisson inverse problems, on both deconvolution and \glsxtrlong{pet} problems.
We implement all the methods and perform our experiments using the \texttt{DeepInverse} Python library \cite{tachella2025deepinverse}.
The \glsxtrshort{pet} experiments make use of the \texttt{parallelproj} projectors \cite{schrammPARALLELPROJOpensourceFramework2024a}.

\subsection{Algorithms}
\label{subsec:algorithms}

For the Tikhonov regularization $\ell_2(x) = \frac{\lambda}{2} \| x \|_2^2$, with $\lambda > 0$, the right Bregman proximal map with respect to the weighted negative entropy is available in closed form:
\begin{equation}
    \overrightarrow{\operatorname{prox}}^\varphi_{\ell_2}(x) = \frac{-s + \sqrt{s^2 + 4 \lambda s \odot x}}{2 \lambda}.
\end{equation}
Therefore the regularized multiplicative updates \eqref{eq:regularized-multiplicative-updates} amount to:
\begin{equation}
  x^{(k+1)} = \frac{-s + \sqrt{s^2 + 4 \lambda x^{(k)} \odot A^\top \frac{y}{Ax^{(k)}} }}{2 \lambda}. \label{eq:mu-l2-algo}
\end{equation}

For the sparsity-promoting regularization $\ell_1(x) = \lambda \|x\|_1$, the right Bregman proximal map is likewise available in closed form for $x \in \mathbb{R}^n_{++}$:
\begin{equation}
    \overrightarrow{\operatorname{prox}}^\varphi_{\ell_1}(x) = \frac{s \odot x}{s + \lambda \mathbf{1}_n}
\end{equation}
which gives
\begin{equation}
    x^{(k+1)} = \frac{x^{(k)}}{s + \lambda\mathbf{1}_n} \odot A^\top \frac{y}{Ax^{(k)}}. \label{eq:mu-l1-algo}
\end{equation}

We finally consider the isotropic \glsxtrfull{tv} regularization, which promotes piecewise-constant solutions \cite{rudinNonlinearTotalVariation1992b}:
\begin{equation}
  \text{\glsxtrshort{tv}}(x) = \lambda \| D x \|_{2,1},
\end{equation}
where $D$ is the discrete image gradient and $d$ denotes the number of spatial dimensions. In our experiments, we use the forward finite differences of \texttt{DeepInverse}, with zero differences at the upper boundary of each spatial axis, and their exact adjoint $D^\top$.
Similarly to the Euclidean setting, the right Bregman proximal map of the \glsxtrshort{tv} regularization does not have a closed-form expression.
For a fixed $z\in\mathbb{R}^n_{++}$, and after discarding terms independent of the optimization variable, the proximal subproblem is
\begin{equation}
  \overrightarrow{\operatorname{prox}}^\varphi_{\text{\glsxtrshort{tv}}}(z) = \argmin_{u \in \mathbb{R}^n_{+}} \left\{ \lambda \|D u\|_{2,1} + \langle s, u - z \odot \log u\rangle \right\}. \label{eq:tv-poisson-energy}
\end{equation}
We solve \eqref{eq:tv-poisson-energy} with the \glsxtrfull{pdhg} algorithm \cite{chambolleFirstOrderPrimalDual2011}.
Let
\begin{equation}
  F(u) = \langle s, u-z\odot\log u\rangle, \qquad G(q) = \lambda\|q\|_{2,1}.
\end{equation}
The Euclidean proximal map of $F$ is separable and has the closed form
\begin{equation}
  \operatorname{prox}_{\tau F}(v) = \frac{v-\tau s+\sqrt{(v-\tau s)^2+4\tau s\odot z}}{2}, \label{eq:tv-pdhg-primal-prox}
\end{equation}
while $\operatorname{prox}_{\sigma G^*}$ is the projection $\Pi_{\mathcal{P}_\lambda}$ onto $\mathcal{P}_\lambda=\{p \in \mathbb{R}^{n\times d} \mid \|p_i\|_2\leq\lambda, \, i \in \{1, \dots, n\} \}$.
Starting from $u^{(0)}=z$, $p^{(0)}=0$, and $\bar u^{(0)}=u^{(0)}$, our inner iterates are
\begin{equation}
  \begin{cases}
    p_i^{(l+1)} = \dfrac{p_i^{(l)}+\sigma(D\bar u^{(l)})_i}{\max\left\{1,\left\|p_i^{(l)}+\sigma(D\bar u^{(l)})_i\right\|_2/\lambda\right\}}, \\[2em]
    u^{(l+1)} = \dfrac{u^{(l)}-\tau D^\top p^{(l+1)}-\tau s+\sqrt{\left(u^{(l)}-\tau D^\top p^{(l+1)}-\tau s\right)^2+4\tau s\odot z}}{2}, \\[1em]
    \bar u^{(l+1)} = 2u^{(l+1)}-u^{(l)}.
  \end{cases}
  \label{eq:mu-tv-pdhg}
\end{equation}

$F$ and $G$ are proper, \glsxtrshort{lsc}, and convex, and the subproblem admits a saddle point.
Consequently, the inner iterates converge to the unique minimizer of \eqref{eq:tv-poisson-energy} whenever $\tau\sigma\|D\|^2<1$ \cite{chambolleFirstOrderPrimalDual2011}.
In our implementation, for $d=2$ spatial dimensions, we use $\tau=\sigma=0.99/\sqrt{4d}$ together with the bound $\|D\|^2\leq4d$.
Other methods for minimizing the functional \eqref{eq:tv-poisson-energy} exist \cite{maximTomographicReconstructionPoisson2023}.
We have chosen \textsc{pdhg} for its simplicity and because, unlike other existing methods, it does not impose an upper bound on the regularization weight $\lambda$.

\subsection{Non-blind Poisson deconvolution}
\label{subsec:deconvolution}

In this experiment the goal is to recover a signal $x \in \mathbb{R}^n_+$ from noisy measurements $y \in \mathbb{R}^n_+$ of the form $y \sim \mathcal{P}(A x)$, where $A$ is a convolution operator, with $A^\top \mathbf{1}_n = \mathbf{1}_n$.
Although the convergence analysis of section \ref{subsec:convergence-analysis-poisson} assumes strictly positive observations $y\in\mathbb{R}^m_{++}$, in this experiment we allow observations that may contain zero entries, i.e. $y\in\mathbb{R}^m_+$.

\begin{figure}[h]
  \centering
  \includegraphics[width=0.7\textwidth]{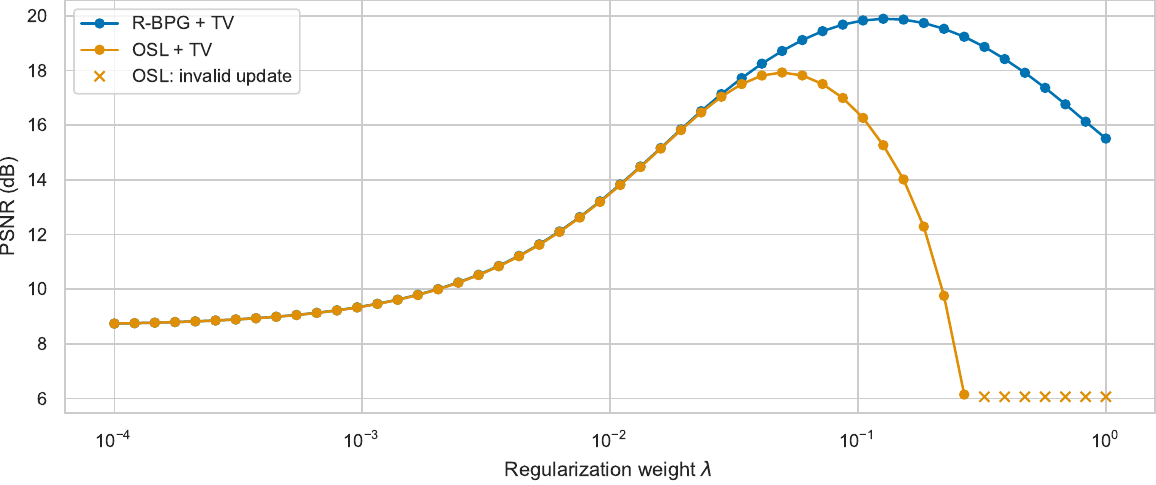}%
  \caption{\Glsxtrshort{psnr} as a function of the regularization weight for \texttt{R-BPG+TV} and \texttt{\glsxtrshort{osl}-\glsxtrshort{tv}} in the sweep of figure \ref{fig:deconvolution-tv-stability-reconstructions}. Crosses along the bottom mark \glsxtrshort{osl} runs that have diverged and have no assoaciated PSNR values.}
  \label{fig:deconvolution-tv-stability-metrics}
\end{figure}

We choose a Gaussian blur kernel with a standard deviation of 2 pixels and circular padding, and we take as ground-truth image the \texttt{butterfly} image from the \texttt{Set3C} dataset, converted to grayscale and resized to $256\times256$ pixels.
We also introduce a gain parameter $\gamma$ so that $y \sim \gamma \mathcal{P}(\frac{A x}{\gamma})$, which enables us to increase the noise level in the measurements. 
We use $\gamma = 10^{-2}$, corresponding to a measurement \glsxtrfull{snr} of approximately 17.52 dB.
We also report a comparison at the higher noise level $\gamma=10^{-1}$ in appendix \ref{app:high-noise-deconvolution}.

\begin{figure}[h]
  \centering
  \includegraphics[width=0.6\textwidth]{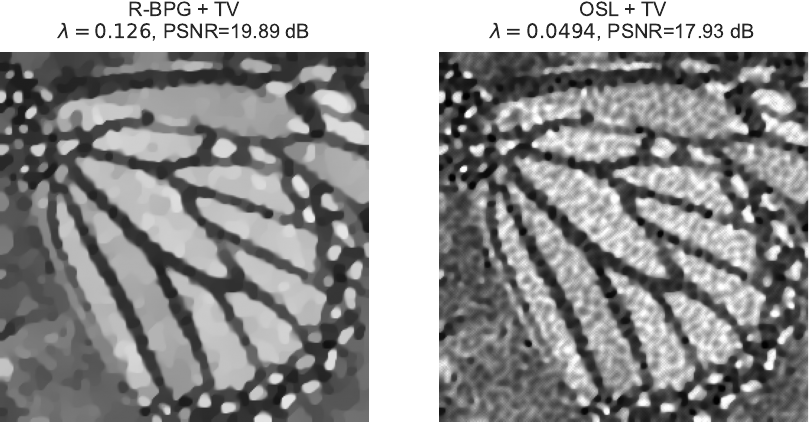}
  \caption{Reconstructions at the best tested weights from the regularization-weight sweep with gain $\gamma=10^{-1}$ and Gaussian blur standard deviation 2 pixels: \texttt{R-BPG+TV} uses $\lambda\approx0.126$ and \texttt{\glsxtrshort{osl}-\glsxtrshort{tv}} uses $\lambda\approx0.0494$.}
  \label{fig:deconvolution-tv-stability-reconstructions}
\end{figure}

We denote by \texttt{\glsxtrshort{mu}} the \glsxtrlong{mu} \eqref{eq:multiplicative-updates} without regularization, by \texttt{R-BPG+L2} the regularized multiplicative updates \eqref{eq:regularized-multiplicative-updates} with Tikhonov regularization, and \texttt{R-BPG+TV} the regularized multiplicative updates with \glsxtrshort{tv} regularization, whose right Bregman proximal step is approximated using the \textsc{pdhg} iterates \eqref{eq:mu-tv-pdhg}.
We compare these algorithms against \texttt{Mirror descent} with Burg's entropy as potential \cite{bauschkeDescentLemmaLipschitz2017}, and its Tikhonov-regularized variant \texttt{BPG+L2}.
For \texttt{Mirror descent}, we use Armijo backtracking, starting each iteration with a trial step size $\tau=1$ and halving it until the candidate $x_\tau$ is strictly positive and satisfies
\begin{equation}
  \Phi(x_\tau)\leq\Phi(x^{(k)})+c\left[\langle\nabla f(x^{(k)}),x_\tau-x^{(k)}\rangle+g(x_\tau)-g(x^{(k)})\right],
  \qquad c=10^{-4},
  \label{eq:bpg-l2-armijo}
\end{equation}
This search allows at most 50 trial steps per iteration, recomputing the proximal candidate at every trial step. 
Both \texttt{MU}, \texttt{R-BPG-L2} and \texttt{R-BPG-TV} use a unit step size.
We also compare against a popular heuristic for Poisson inverse problems sometimes called Green's algorithm or \glsxtrshort{osl} \cite{greenUseEmAlgorithm1990,greenBayesianReconstructionsEmission1990}, which we covered at the end of section \ref{sec:regularized-scheme}.  
The Tikhonov and \glsxtrshort{tv} regularized \glsxtrshort{osl} updates \eqref{eq:osl-multiplicative-updates} \cite{deyRichardsonLucyAlgorithm2006}, denoted by \texttt{OSL+L2} and \texttt{OSL+TV}.
For \texttt{R-BPG+TV}, we run 100 inner \textsc{pdhg} iterations at each outer iteration, restarting the inner solver from the current multiplicative update with a zero dual variable. 
All compared methods are initialized with an image of ones. 

\begin{figure}[h]
  \centering
  \includegraphics[width=0.8\textwidth]{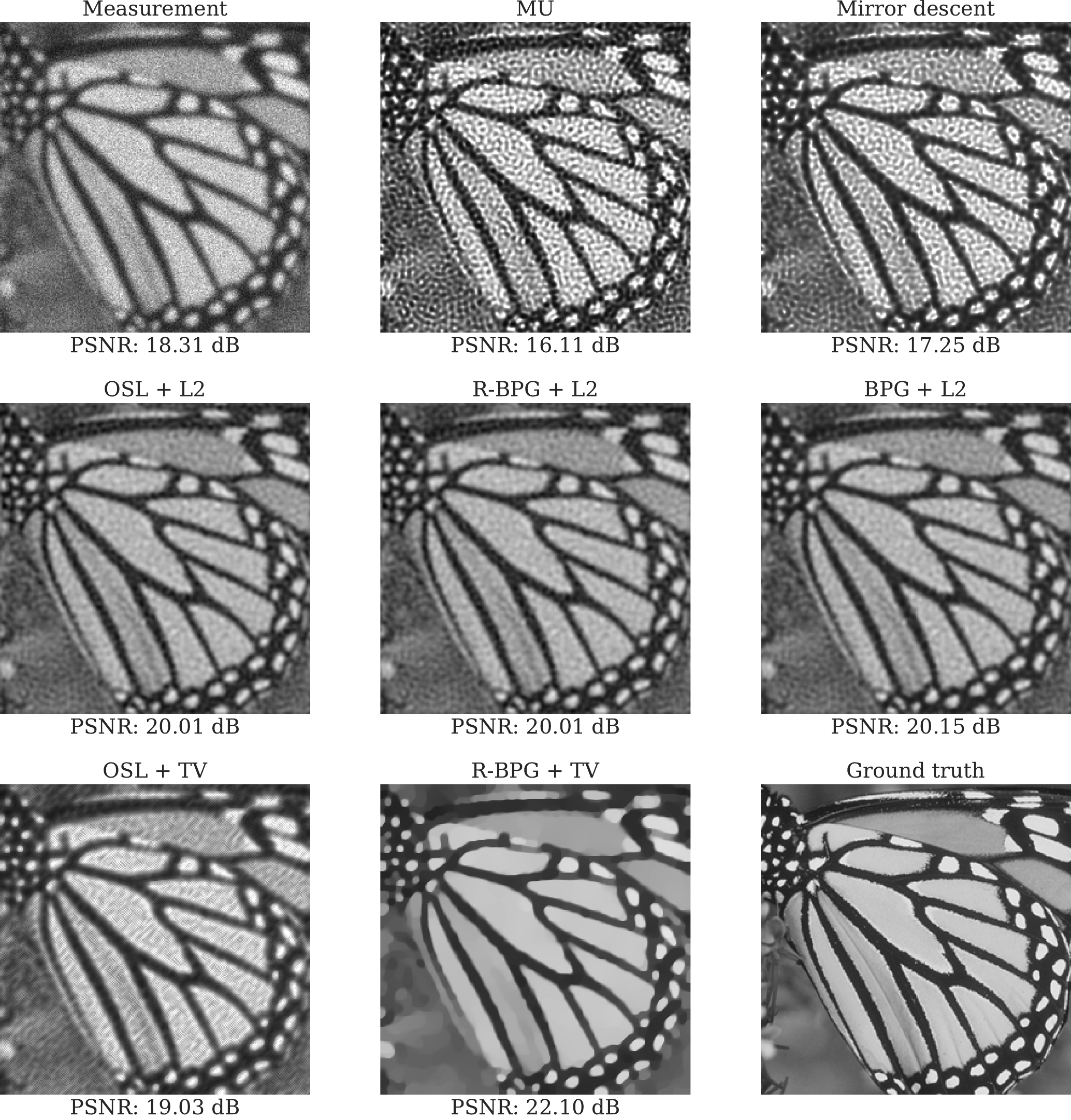}
  \caption{Non-blind deconvolution of the \texttt{butterfly} image from the \texttt{Set3C} dataset, with gain $\gamma = 10^{-2}$. Tikhonov regularization uses $\lambda=0.15$, and \glsxtrshort{tv} regularization uses $\lambda=0.08$.}
  \label{fig:deconvolution-reconstructions-1}
\end{figure}

For the TV stability comparison, we sweep 50 logarithmically spaced regularization weights between $10^{-4}$ and $1$, with 100 outer iterations per reconstruction.
We report the \glsxtrshort{psnr} values obtained for all the tested reconstruction parameters for \texttt{R-BPG+TV} and \texttt{OSL+TV} in figure \ref{fig:deconvolution-tv-stability-metrics}, in the deconvolution setting with a gain $\gamma = 10^{-1}$.
When the regularization weight is small, both methods produce similar reconstructions and achieve similar \glsxtrshort{psnr} values.
For larger weights, the reconstruction quality of \texttt{OSL+TV} deteriorates. The sweep aborts an \glsxtrshort{osl} run if its denominator becomes nonpositive or an iterate becomes nonfinite; here, all seven aborted runs encounter a nonpositive denominator, starting at the tested weight $\lambda\approx0.324$.
The \texttt{R-BPG+TV} iterates remain finite and nonnegative for all tested regularization weights.
The best tested weights are $\lambda\approx0.126$ for \texttt{R-BPG+TV} and $\lambda\approx0.0494$ for \texttt{OSL+TV}, with \glsxtrshort{psnr} values of 19.89 dB and 17.93 dB, respectively.
The corresponding reconstructions are shown in figure \ref{fig:deconvolution-tv-stability-reconstructions}.

We compare all the mentioned algorithms at the noise level $\gamma = 10^{-2}$ and plot the reconstructions in figure \ref{fig:deconvolution-reconstructions-1}.
All the algorithms are run for 150 iterations.
We report the selected regularization parameters in the captions of the figures.
We also plot the evolution of the different composite objectives, as well as the evolution of the \glsxtrfull{psnr} and the relative distance between iterates in figure \ref{fig:deconvolution-reconstructions-metrics-1}.
The \texttt{MU} and \texttt{Mirror descent} algorithms without regularization produce typical noise amplification artifacts.
Adding Tikhonov regularization mitigates these artifacts, but at the cost of blurry reconstructions.
With this choice of regularization, the \texttt{OSL+L2} method is stable and produces a reconstruction of similar quality to \texttt{R-BPG+L2}.
The \texttt{R-BPG+TV} algorithm mitigates noise artifacts while preserving edges.
At this regularization weight, the \texttt{OSL+TV} method exhibits persistent iterate oscillations and visible reconstruction artifacts.

\begin{figure}[h]
  \centering
  \includegraphics[width=\textwidth]{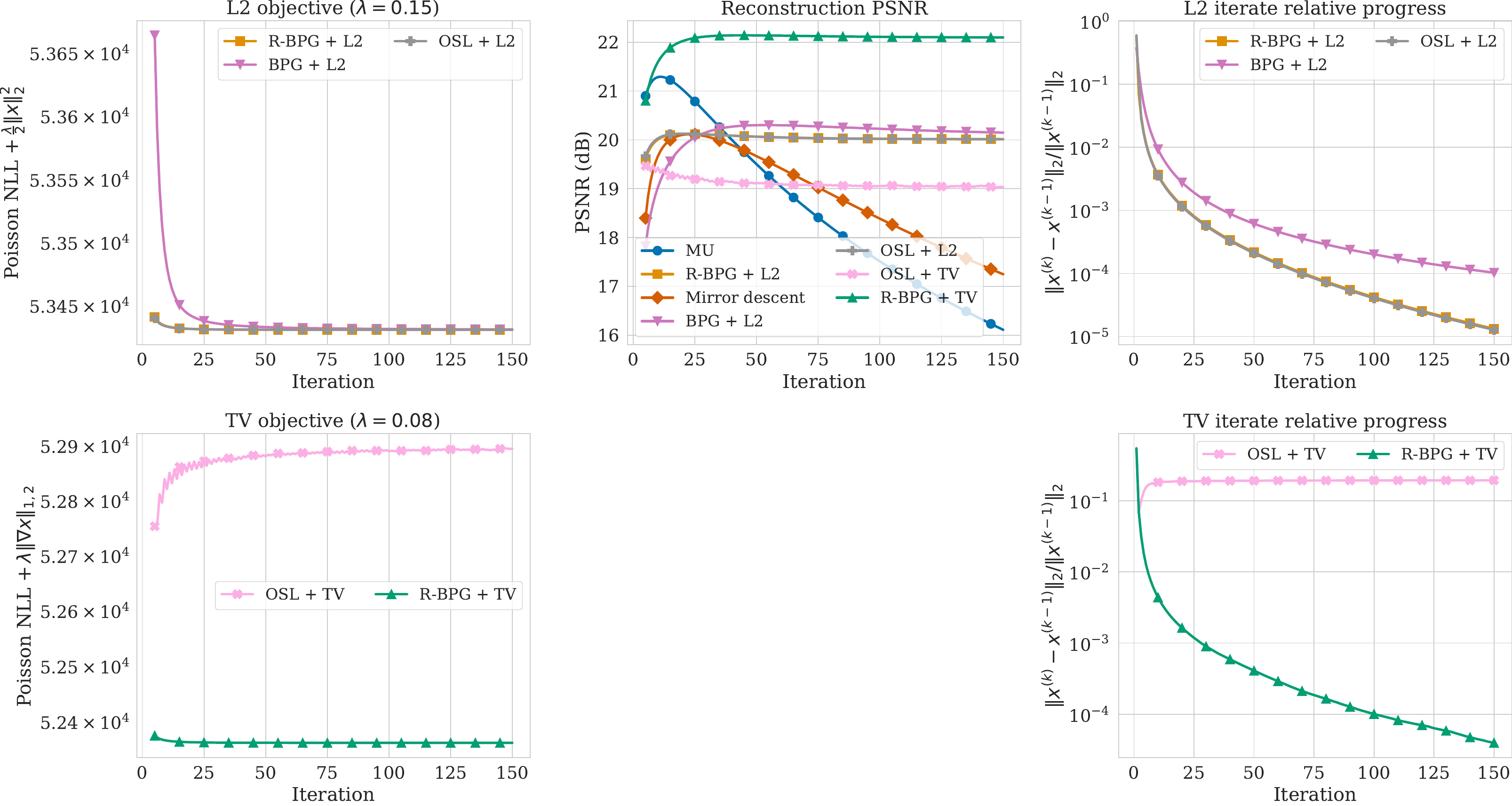}
  \caption{Composite objectives, \glsxtrshort{psnr} in dB, and relative distances between successive iterates for the non-blind deconvolution experiment with gain $\gamma = 10^{-2}$.}
  \label{fig:deconvolution-reconstructions-metrics-1}
\end{figure}

\subsection{Positron emission tomography}

We consider two tomographic reconstruction experiments, in 2D and 3D, for \glsxtrshort{pet}.
The forward model of \glsxtrshort{pet} is a composition of several linear degradations.
It can be described as recovering an image $x \in \mathbb{R}^n_+$ from measurements $y \in \mathbb{R}^m_+$ of the form:
\begin{equation}
  \label{eq:pet-problem}
  y \sim \gamma \mathcal{P}\left( \frac{ARHx + b}{\gamma} \right)
\end{equation}
where $A$ is an attenuation matrix, often representing an additional X-ray \glsxtrfull{ct} scan of the patient, $R$ is the discretized Radon transform and models the geometry of the scanner, $H$ is a convolution operator representing the \glsxtrfull{psf} of the system, and $b \in \mathbb{R}^m_{++}$ is a background term representing the contribution of random and scattered events.
Similarly to the deconvolution experiment, the gain parameter $\gamma$ controls the noise level. 

\begin{figure}[h]
  \centering
  \includegraphics[width=\textwidth]{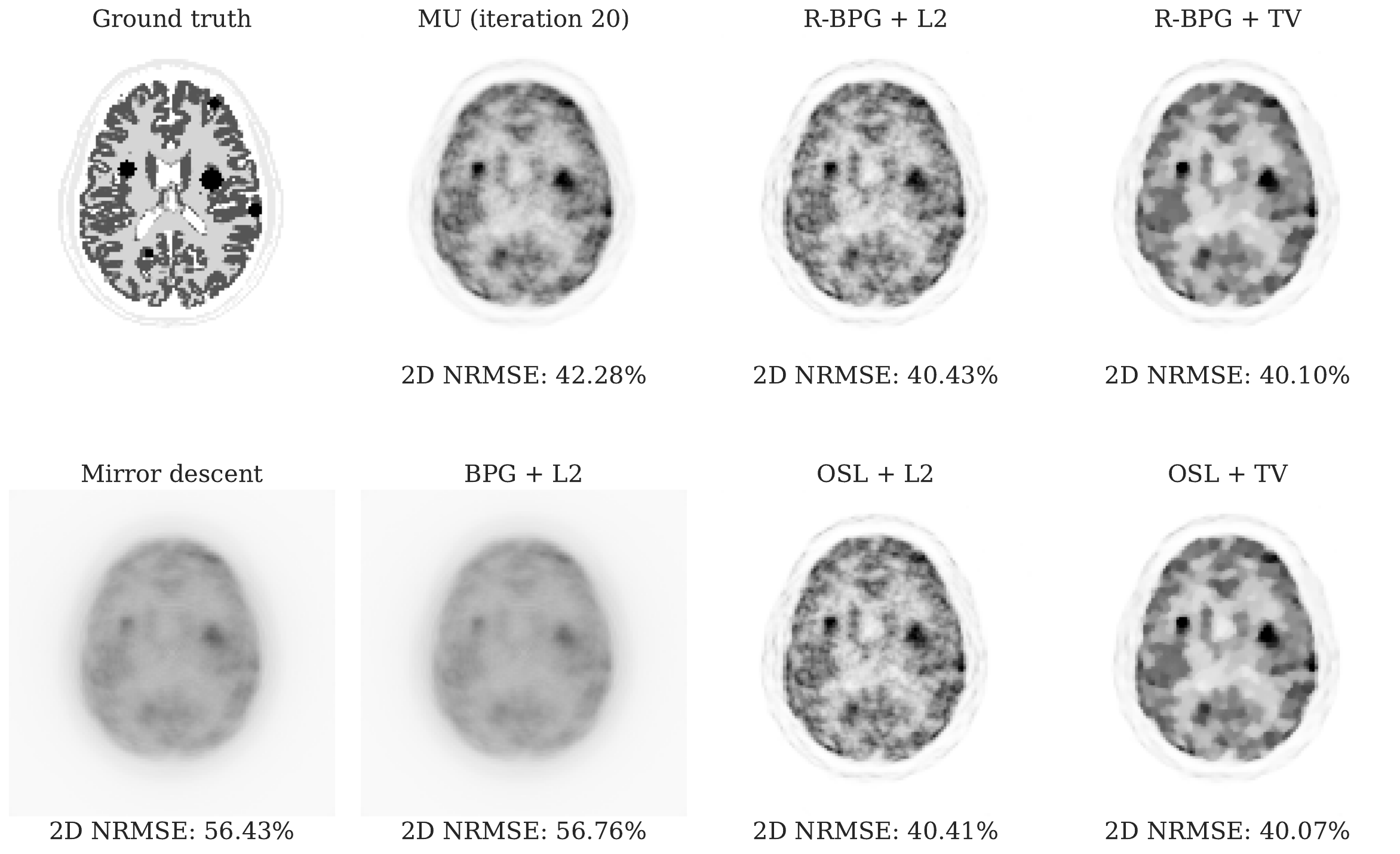}
  \caption{Reconstructions of the $120\times120$ 2D \texttt{BrainWeb} \glsxtrshort{pet} phantom, with 2 mm pixel spacing, 3 mm \glsxtrshort{psf} full-width at half maximum, and $500{,}000$ expected prompt coincidences.}
  \label{fig:tomographic-reconstructions}
\end{figure}

Both experiments use subject 04 from \texttt{BrainWeb} \cite{collinsDesignConstructionRealistic1998}.
The FDG activity and attenuation maps are obtained from the open-source \texttt{brainweb} implementation \cite{daCostaLuisBrainweb2020}.
We insert five hot lesions with diameters 6, 8, 10, 12, and 16 mm and activity 192, which is 50\% above the FDG grey-matter activity of 128, then normalize the emission image by its maximum.
For both acquisitions, the Gaussian \glsxtrshort{psf} has a full-width at half maximum of 3 mm. 
We use an acquisition geometry with 34 polygon sides, eight detector bins per side, and a spacing of 8 mm. 
We target $N=500{,}000$ counts, which is equivalent to again of approximately $\gamma\approx1.33\times10^{-3}$.

We compare the same seven algorithms as in the deconvolution experiment. 
We early-stop \texttt{MU} at 20 iterations and run every other method for 50 iterations. 
The Tikhonov weight is $\lambda=0.01$ in both dimensions while the TV weight is $\lambda=0.0015$. 
For \texttt{R-BPG+TV}, we run 100 inner \textsc{pdhg} iterations at every outer iteration.

The reconstructions are shown in figure \ref{fig:tomographic-reconstructions}, and the composite objectives, \glsxtrshort{nrmse}, and relative distances between successive iterates are shown in figure \ref{fig:tomographic-reconstructions-metrics}.
We also report the \glsxtrfull{rc} for each lesion, which is the ratio of mean reconstructed activity to mean ground-truth activity within each labeled lesion.

\begin{figure}[h]
  \centering
  \includegraphics[width=\textwidth]{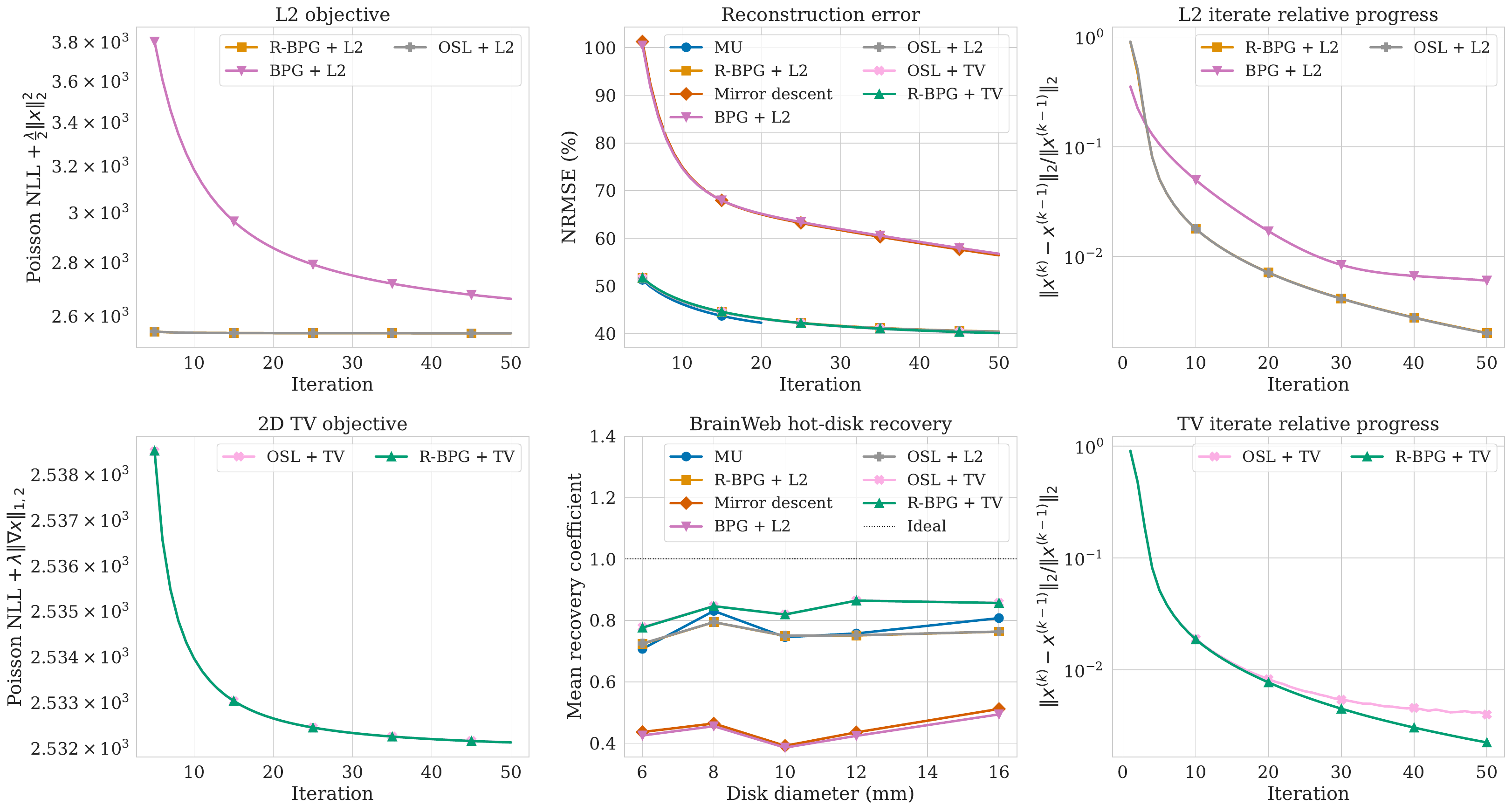}
  \caption{Composite objectives, \glsxtrshort{nrmse} in percent, relative distances between successive iterates and recovery coefficients for the 2D \texttt{BrainWeb} \glsxtrshort{pet} experiment.}
  \label{fig:tomographic-reconstructions-metrics}
\end{figure}

\section{Conclusion and discussions}
\label{sec:discussions}

We introduced \glsxtrshort{rbpg}, a novel algorithm to minimize convex composite objective functions, based on a reparameterization of the problem in mirror coordinates.
This approach makes it possible to use potential functions that are not covered by standard relative smoothness.
In the Poisson case \glsxtrshort{rbpg} reduces to a regularized version of the well-known multiplicative updates algorithm.
Our numerical results suggest that in the Poisson setting, the proposed algorithm is competitive with existing methods such as \texttt{BPG}, using unit-step multiplicative updates and allowing general convex regularization.

The general convergence analysis of \glsxtrshort{rbpg} remains incomplete.
In particular establishing convergence rates and convergence of the iterates in the general setting is a natural direction for future work.
Efficient computations of the Bregman proximal map in the case of the weighted negative entropy for other convex regularization is also of interest for the method to be applicable to larger scale problems.
Another interesting extension is the case where $f$, $g$, or both, are nonconvex.
Allowing for non-convexity would open the door to analyzing the behavior of the algorithm when deep neural networks are used as regularizers.
Finally we conjecture that the convergence rate in the Poisson case can be improved for some cases, as numerical experiments show a decrease of the objective faster than the sublinear bound established in theorem \ref{theorem:function-value-rate}.

\clearpage

\appendix

\section{Proofs}
\label{appendix:proofs}

\subsection{Failure of relative smoothness with negative entropy}
\label{proof:lem:non-relative-smoothness-poisson}

\begin{proof}
  Let $v=\mathbf{1}_n$ and $x=t\mathbf{1}_n$ with $t>0$, so that $x$ lies
  in the interior of the positive orthant. Since $Av>0$,
  \begin{equation*}
    v^\top\nabla^2 f(x)v=\frac{\|y\|_1}{t^2},
    \qquad v^\top\nabla^2\varphi(x)v=\frac{n}{t},
  \end{equation*}
  for $\varphi(x)=\sum_i(x_i\log x_i-x_i)$.
  Relative smoothness would therefore require $\|y\|_1\leq Lnt$ for every
  $t>0$, which is impossible as $t\downarrow0$.
\end{proof}

\subsection{Proof of lemma \ref{lem:relative-smoothness-reparameterized}}
\label{proof:lem:relative-smoothness-reparameterized}

\begin{proof}
  We use the characterization of proposition \ref{prop:mirror-relative-smoothness-characterization}.
  Fix $x \in \intdom \varphi = \mathbb{R}_{++}^n$.
  With $s=A^\top\mathbf{1}_m$, recall the derivatives of the Poisson \glsxtrshort{nll} are
  \begin{equation*}
    \nabla f(x)=s-A^\top\frac{y}{Ax}
    \quad\text{and}\quad
    \nabla^2 f(x)=A^\top \operatorname{diag}\left(\frac{y}{(Ax)^2}\right)A.
  \end{equation*}
  The derivatives of the weighted negative entropy are given by:
  \begin{equation}
    \label{eq:weighted-negentropy-third-derivative}
    \nabla^2 \varphi(x) = \operatorname{diag}\left( \frac{s}{x}\right)~\text{and}~\nabla^3 \varphi(x) [u] = -\operatorname{diag} \left( \frac{s \odot u}{x^2}\right).
  \end{equation}
  Thus replacing every term by its expression:
  \begin{equation*}
    \left[\nabla^2\varphi(x)\right]^{-1}\nabla f(x)
    =\frac{x}{s}\odot\left(s-A^\top\frac{y}{Ax}\right)
    =x-\frac{x}{s}\odot A^\top\frac{y}{Ax},
  \end{equation*}
  and:
  \begin{equation*}
    \nabla^3\varphi(x)
    \left[\left[\nabla^2\varphi(x)\right]^{-1}\nabla f(x)\right]
    =-\operatorname{diag}\left(\frac{s}{x}-\frac{1}{x}\odot A^\top\frac{y}{Ax}\right).
  \end{equation*}
  Substituting these expressions into \eqref{eq:mirror-relative-smoothness-characterization} yields:
  \begin{align*}
    &\nabla^2 f(x)
    -\nabla^3\varphi(x)
    \left[\left[\nabla^2\varphi(x)\right]^{-1}\nabla f(x)\right] \\
    &\qquad = \operatorname{diag}\left(\frac{s}{x}\right)
    +A^\top \operatorname{diag}\left(\frac{y}{(Ax)^2}\right)A
    -\operatorname{diag}\left(\frac{1}{x}\odot A^\top\frac{y}{Ax}\right).
  \end{align*}
  It remains to show that the last two terms form a negative semidefinite matrix.
  For any $w\in\mathbb{R}^n$, the Cauchy--Schwarz inequality and the nonnegativity of the coefficients of $A$ give, for every $1 \leq i \leq m$:
  \begin{equation}
    \label{eq:cauchy-schwarz}
    \left(\sum_{j=1}^n a_{ij}w_j\right)^2
    \leq
    \left(\sum_{j=1}^n a_{ij}x_j\right)
    \left(\sum_{j=1}^n \frac{a_{ij}}{x_j}w_j^2\right)
    =(Ax)_i\sum_{j=1}^n \frac{a_{ij}}{x_j}w_j^2.
  \end{equation}
  Hence
  \begin{align*}
    w^\top A^\top \operatorname{diag}\left(\frac{y}{(Ax)^2}\right)Aw
    &=\sum_{i=1}^m \frac{y_i}{(Ax)_i^2}
      \left(\sum_{j=1}^n a_{ij}w_j\right)^2 \\
    &\leq \sum_{i=1}^m \frac{y_i}{(Ax)_i}
      \sum_{j=1}^n \frac{a_{ij}}{x_j}w_j^2 \\
    &=\sum_{j=1}^n \frac{1}{x_j}\left(A^\top\frac{y}{Ax}\right)_j w_j^2 \\
    &=w^\top\operatorname{diag}\left(\frac{1}{x}\odot A^\top\frac{y}{Ax}\right)w,
  \end{align*}
  and:
  \begin{equation*}
    \nabla^2 f(x)
    -\nabla^3\varphi(x)
    \left[\left[\nabla^2\varphi(x)\right]^{-1}\nabla f(x)\right]
    \preceq \operatorname{diag}\left(\frac{s}{x}\right)
    =\nabla^2\varphi(x).
  \end{equation*}
  Proposition \ref{prop:mirror-relative-smoothness-characterization} then shows that
  $f\circ\nabla\varphi^*$ is $1$-smooth relative to $\varphi^*$.

  Now assume that $L < 1$.
  Taking $w=x$ in \eqref{eq:cauchy-schwarz} turns the Cauchy--Schwarz inequality into an equality, and substituting it into the mirror relative smoothness condition \eqref{eq:mirror-relative-smoothness-characterization} gives:
  \begin{equation*}
    x^\top \nabla^2 \varphi(x) x \leq L x^\top \nabla^2 \varphi(x) x,
  \end{equation*}
  which is a contradiction, hence $L=1$.
\end{proof}

\subsection{Non-convexity of the reparameterized objective}
\label{appendix:subsec:non-convexity-reparameterized}

We provide a simple counter-example showing that the reparameterized objective $f \circ \nabla \varphi^*$ is not convex in general, along with a visual example in Figure \ref{fig:non-convex}.
Let us consider the Poisson \glsxtrshort{nll} with $A = (1, 1)$ and $y = 1$, and the reparameterization induced by the choice of potential $\varphi^* (\eta) = \sum_{i=1}^2 \exp(\eta_i)$. Note that this matches the assumptions described in the beginning of section \ref{subsec:convergence-analysis-poisson}.
The function has the following expression:
\begin{equation*}
  f \circ \nabla \varphi^* (\eta) = e^{\eta_1} + e^{\eta_2} - \log \left( e^{\eta_1} + e^{\eta_2} \right).
\end{equation*}
Then the Hessian of this function is given by:
\begin{equation}
  \nabla^2 (f \circ \nabla \varphi^*)(\eta) =  \begin{pmatrix}
    e^{\eta_1} - \frac{e^{\eta_1 + \eta_2}}{(e^{\eta_1} + e^{\eta_2})^2} & \frac{e^{\eta_1 + \eta_2}}{(e^{\eta_1} + e^{\eta_2})^2} \\
    \frac{e^{\eta_1 + \eta_2}}{(e^{\eta_1} + e^{\eta_2})^2} & e^{\eta_2} - \frac{e^{\eta_1 + \eta_2}}{(e^{\eta_1} + e^{\eta_2})^2}
  \end{pmatrix}
\end{equation}
now consider the point $\eta = (t, t)$ for $t \in \mathbb{R}$, then the Hessian at this point is given by:
\begin{equation*}
  \nabla^2 (f \circ \nabla \varphi^*)(t, t) =  \begin{pmatrix}
    e^t - \frac{1}{4} & \frac{1}{4} \\
    \frac{1}{4} & e^t - \frac{1}{4}
  \end{pmatrix}
\end{equation*}
and its eigenvalues are given by $e^t$ and $e^t - \frac{1}{2}$, which means that the Hessian is not positive semidefinite for $t < - \log (2)$, hence the function is not convex on this region of the diagonal $\eta = (t, t)$.

\begin{figure}[h]
  \centering
  \includegraphics[width=0.7\textwidth]{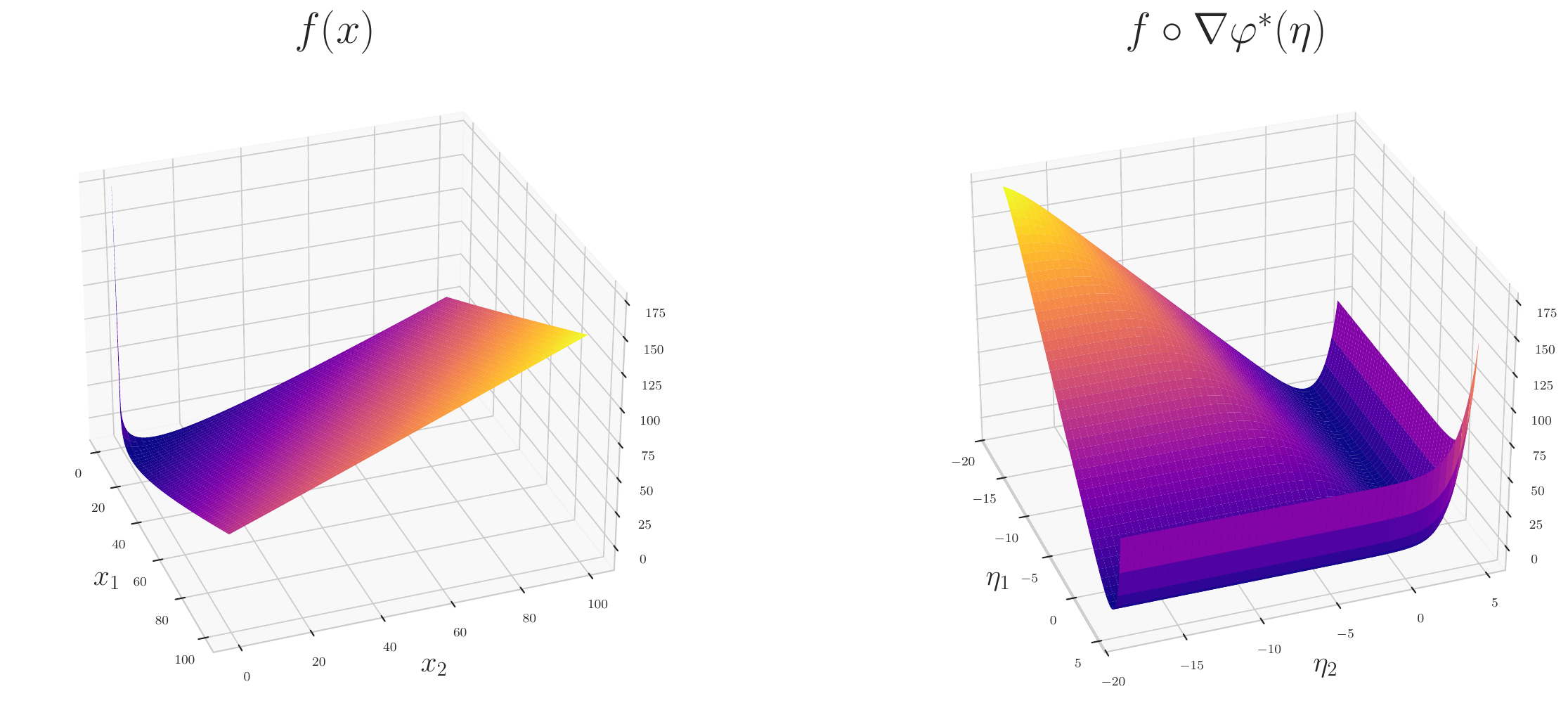}
  \caption{Plot of the function Poisson \glsxtrshort{nll} $f$ and its reparameterization $f \circ \nabla \varphi^*$, in the case where $A = (1, 1)$ and $y = 10$. For the plotted value $y=10$, the Hessian has a negative eigenvalue at $\eta=(t,t)$ whenever $t<\log(5)$.}
  \label{fig:non-convex}
\end{figure}

\section{Convergence analysis using non-convex \texorpdfstring{\glsxtrshort{bpg}}{BPG}}
\label{appendix:sec:non-convex}\label{appendix:non-convex-analysis}

This section presents a complementary convergence analysis of \glsxtrshort{rbpg} in mirror coordinates.
Although $\Phi$ is convex, the nonlinear change of coordinates induced by the mirror map does not generally preserve convexity.
The reparameterized objective $\Phi \circ \nabla \varphi^*$ is thus generally non-convex, and the convergence theory for convex \glsxtrshort{bpg} cannot be applied directly.

We therefore rely on the non-convex \glsxtrshort{bpg} analysis of \cite{bolteFirstOrderMethods2018}, which follows the general convergence framework of \cite{attouchConvergenceProximalAlgorithm2009}.
This approach provides convergence guarantees for the iterates after verifying a collection of regularity, coercivity, and geometric assumptions in mirror coordinates.
These assumptions are substantially more restrictive in the Poisson setting, particularly with respect to the regularization term $g$.
This section thus aims at highlighting these restrictions and comparing the resulting guarantees with those obtained from the primal analysis of section \ref{sec:convergence-analysis-convex}.

\subsection{Assumptions}
\label{appendix:sec:non-convex:subsec:assumptions}

We retain the convexity and regularity requirements on $f$ and $\varphi$ in
assumption \ref{assumption:blanket}, the lower bound on $\Phi$, and assumption
\ref{assumption:dual-relative-smoothness}.
For this alternative analysis, we allow a proper, \glsxtrshort{lsc}, convex
regularizer $g$ whose domain need only intersect $\intdom\varphi$; the inclusion
$\operatorname{dom}\varphi\subset\operatorname{dom}g$ is not required.
This admits the boundary-penalizing regularizers discussed below.
We assume $\intdom\varphi^*=\mathbb R^n$, so that
$g\circ\nabla\varphi^*$ is proper and \glsxtrshort{lsc} on the whole mirror space.
We initialize in $\operatorname{dom}g\cap\intdom\varphi$ and use the mirror
minimization formulation \eqref{eq:right-bregman-proximal-gradient-minimization-eta}.
When the natural-gradient point lies in $\intdom\varphi$, this is equivalently
\eqref{eq:bpg-right}.
To ensure that the left Bregman proximal problem defining $\leftprox^{\varphi^*}_{\tau g \circ \nabla \varphi^*}$ has a nonempty compact solution set within $\intdom \varphi^*$ we impose the following additional condition.
\begin{assumption}[Supercoercivity in mirror coordinates]
  \label{assumption:dual-supercoercivity}
  For all $\tau > 0$, the function $\varphi^* + \tau g \circ \nabla \varphi^*$ is supercoercive on $\intdom \varphi^*$, that is,
  \begin{equation*}
    \lim_{\| \eta \| \to +\infty} \frac{\varphi^*(\eta) + \tau g \circ \nabla \varphi^*(\eta)}{\| \eta \|} = +\infty.
  \end{equation*}
\end{assumption}
This condition corresponds to \cite[assumption B]{bolteFirstOrderMethods2018}.
The proximal map can be set-valued when the minimizer is not unique.
In that case, the algorithm selects any minimizer at each iteration.

The full-domain assumption ensures that the proximal iterates remain in $\intdom\varphi^*$, as required in \cite[assumption C]{bolteFirstOrderMethods2018}.

These assumptions provide descent and asymptotic regularity properties for the \glsxtrshort{bpg} iterates \cite[proposition 4.1]{bolteFirstOrderMethods2018}.
Convergence of subsequences requires the following additional local regularity assumptions.
\begin{assumption}[Local regularity assumptions]
  \label{assumption:subsequence-convergence}
  The following properties hold:
  \begin{enumerate}
    \item $\varphi^*$ is twice continuously differentiable and strongly convex on every bounded subset of $\intdom \varphi^*$.
    \item $\nabla \varphi^*$ and $\nabla (f \circ \nabla \varphi^*)$ are Lipschitz continuous on every bounded subset of $\intdom \varphi^*$.
  \end{enumerate}
\end{assumption}
Global strong convexity of $\varphi^*$ on $\mathbb{R}^n$ is assumed in \cite[assumption D]{bolteFirstOrderMethods2018}.
For the convergence results used here, strong convexity on bounded subsets is sufficient, as observed, for example, in \cite[theorem 12]{huraultConvergentPlugandplayMethods2023}.

Full-sequence convergence requires two further properties.
First, the reparameterized objective $\Phi \circ \nabla \varphi^*$ must satisfy the \glsxtrshort{klp} property.
This property holds for broad classes of functions, including proper \glsxtrshort{lsc} subanalytic functions and, more generally, functions definable in an $o$-minimal structure \cite{bolteLojasiewiczInequalityNonsmooth2007}.
In particular, this condition holds when $\Phi$ and the mirror map $\nabla\varphi^*$ are real analytic on the relevant domains, because a composition of real-analytic functions is real analytic and therefore satisfies the \glsxtrshort{klp} property.

Second, the sequence of mirror iterates $\left\{ \eta^{(k)} \right\}_{k \in \mathbb{N}}$ must remain bounded.
A sufficient condition is the coercivity of the reparameterized objective.
\begin{assumption}
  \label{assumption:full-sequence-convergence}
  The reparameterized objective $\Phi \circ \nabla \varphi^*$ is coercive on $\intdom \varphi^*$, that is,
  \begin{equation*}
    \lim_{\| \eta \| \to +\infty} \Phi \circ \nabla \varphi^* (\eta) = +\infty.
  \end{equation*}
\end{assumption}
This assumption is different from coercivity of $\Phi$ in the original variable.
Indeed, under the change of variables $\eta = \nabla \varphi(x)$, a bounded sequence in primal coordinates need not map to a bounded sequence in mirror coordinates.

\subsection{Convergence analysis}

Under these assumptions, the non-convex \glsxtrshort{bpg} analysis of \cite{bolteFirstOrderMethods2018} applies to the reparameterized problem.
We first state the resulting descent and asymptotic regularity properties.
\begin{corollary}[Function-value convergence {\cite[proposition 4.1]{bolteFirstOrderMethods2018}}]
  Let $\left\{ x^{(k)} \right\}_{k \in \mathbb{N}}$ be the sequence generated by \glsxtrshort{rbpg} \eqref{eq:bpg-right} with $0 < \tau < 1/L$.
  Then the sequence of function values $\left\{ \Phi(x^{(k)}) \right\}_{k \in \mathbb{N}}$ is non-increasing and the iterates $\left\{ x^{(k)} \right\}_{k \in \mathbb{N}}$ satisfy
  \begin{equation*}
    \sum_{k=1}^{+\infty} D_{\varphi}(x^{(k)} \parallel x^{(k+1)}) < +\infty.
  \end{equation*}
  Hence,
  \begin{equation*}
    \lim_{k\to+\infty}D_{\varphi}(x^{(k)} \parallel x^{(k+1)})=0.
  \end{equation*}
  Moreover, for every $N\geq 1$,
  \begin{equation}
    \min_{1 \leq k \leq N} D_\varphi (x^{(k)} \parallel x^{(k+1)}) \leq \frac{\tau}{N} \left( \frac{\Phi (x^{(0)}) - \Phi^{\star}}{1 - \tau L}\right)
  \end{equation}
  where $\Phi^\star = \inf_{x \in \intdom\varphi} \Phi(x)$.
\end{corollary}
\begin{proof}
  Proposition 4.1 of \cite{bolteFirstOrderMethods2018} is stated for the reparameterized objective $\Phi \circ \nabla \varphi^*$ and the mirror iterates $\eta^{(k)}$.
  Since $x^{(k)} = \nabla \varphi^*(\eta^{(k)})$, its function-value conclusions apply directly to $\Phi(x^{(k)})$.
  Moreover, item 2 of proposition \ref{prop:legendre-properties} gives
  \begin{equation*}
    D_{\varphi^*}(\eta^{(k+1)} \parallel \eta^{(k)})
    = D_\varphi(x^{(k)} \parallel x^{(k+1)}),
  \end{equation*}
  which yields the stated bounds in primal coordinates.
\end{proof}

This result is weaker than the function-value rate stated for convex \glsxtrshort{bpg} in \cite[theorem 4.1]{teboulleSimplifiedViewFirst2018}.
This is a consequence of analyzing the algorithm in the mirror coordinates $\eta$, where the objective is generally non-convex, even though the original objective $\Phi$ is convex in $x$.

Before establishing convergence of the iterates, we relate critical points of the reparameterized objective to minimizers of the original convex problem.
\begin{lemma}
    \label{lem:critical-to-global}
    Let $\Phi : \mathbb{R}^n \to \mathbb{R}\cup\{+\infty\}$ be proper, \glsxtrshort{lsc} and convex, and let $\varphi$ be of Legendre type. Assume that $\varphi^* \in \mathcal{C}^2(\intdom \varphi^*)$ and that $\nabla^2\varphi^*(\eta)$ is positive definite for every $\eta\in\intdom\varphi^*$.
    Let $\eta^\star \in \intdom \varphi^*$ be a critical point of
    $\Phi \circ \nabla \varphi^*$, in the sense that $0 \in \partial \bigl(\Phi \circ \nabla \varphi^*\bigr)(\eta^\star)$, where $\partial$ denotes the limiting subdifferential \cite{rockafellarVariationalAnalysis1998}.
    Then $x^\star = \nabla \varphi^*(\eta^\star) \in \intdom \varphi$ is a minimizer of $\Phi$, and $\eta^\star$ is a global minimizer of $\Phi \circ \nabla \varphi^*$.
\end{lemma}
\begin{proof}
    Let $x^\star := \nabla \varphi^*(\eta^\star)$.
    Since $\varphi$ is Legendre, the map $\nabla \varphi^*$ is a bijection from $\intdom \varphi^*$ onto $\intdom \varphi$.
    Since $\eta^\star$ is a critical point of $\Phi \circ \nabla \varphi^*$, we have
    \begin{equation*}
        0 \in \partial (\Phi \circ \nabla \varphi^*)(\eta^\star).
    \end{equation*}
    By the chain rule for the limiting subdifferential,
    \begin{equation*}
        \partial (\Phi \circ \nabla \varphi^*)(\eta^\star)
        \subset
        \nabla^2 \varphi^*(\eta^\star)^\top \, \partial \Phi(\nabla \varphi^*(\eta^\star)).
    \end{equation*}
    Since $x^\star = \nabla \varphi^*(\eta^\star)$ and $\nabla^2 \varphi^*(\eta^\star)$ is symmetric, this gives
    \begin{equation*}
        0 \in \nabla^2 \varphi^*(\eta^\star) \, \partial \Phi(x^\star).
    \end{equation*}
    Hence there exists $u^\star \in \partial \Phi(x^\star)$ such that
    \begin{equation*}
        \nabla^2 \varphi^*(\eta^\star)\,u^\star = 0.
    \end{equation*}
    By assumption, $\nabla^2 \varphi^*(\eta^\star)$ is positive definite and hence invertible.
    Thus $u^\star = 0$, and
    \begin{equation*}
        0 \in \partial \Phi(x^\star).
    \end{equation*}
    Since $\Phi$ is convex, this implies that $x^\star$ is a global minimizer of $\Phi$.
    Finally, for every $\eta \in \intdom \varphi^*$, there exists $x \in \intdom \varphi$ such that $x = \nabla \varphi^*(\eta)$. Therefore,
    \begin{equation*}
        (\Phi \circ \nabla \varphi^*)(\eta)
        = \Phi(x)
        \geq \Phi(x^\star)
        = (\Phi \circ \nabla \varphi^*)(\eta^\star).
    \end{equation*}
    Hence $\eta^\star$ is a global minimizer of $\Phi \circ \nabla \varphi^*$.
\end{proof}
A consequence of this lemma is that the reparameterized objective has no spurious critical points: every critical point corresponds to a global minimizer of the original convex objective.
This property is typically associated with invex functions.

We can now transfer the full-sequence convergence result of \cite{bolteFirstOrderMethods2018} to the primal iterates of \glsxtrshort{rbpg}.
\begin{corollary}[Full convergence of the iterates {\cite[Theorem 4.1]{bolteFirstOrderMethods2018}}]
  \label{corollary:full-sequence}
  Suppose the assumptions of this appendix hold, including the
  \glsxtrshort{klp} property of $\Phi\circ\nabla\varphi^*$.
  Let $\left\{ x^{(k)} \right\}_{k \in \mathbb{N}}$ be the sequence generated by \glsxtrshort{rbpg} \eqref{eq:bpg-right} with $0<\tau < 1/L$.
  Then $\left\{ x^{(k)} \right\}_{k \in \mathbb{N}}$ converges to a minimizer of $\Phi$.
\end{corollary}
\begin{proof}
  Theorem 4.1 of \cite{bolteFirstOrderMethods2018} states that the sequence $\left\{ \eta^{(k)} \right\}_{k \in \mathbb{N}}$ converges to a critical point $\eta^\star$ of $\Phi \circ \nabla \varphi^*$.
  Since $\nabla \varphi^*$ is continuous, the sequence $x^{(k)} = \nabla \varphi^*(\eta^{(k)})$
  converges to $x^\star = \nabla \varphi^*(\eta^\star)$, which is a minimizer of $\Phi$ by lemma \ref{lem:critical-to-global}.
\end{proof}

With respect to the iterates, this conclusion is stronger than the guarantees obtained in section \ref{sec:convergence-analysis-convex}, since it establishes convergence of the full sequence to a minimizer.
Rates for the iterates can also be derived from the \glsxtrshort{klp} exponent of the reparameterized objective.
The stronger conclusion comes at the cost of the additional assumptions introduced above.
We next examine their implications in the Poisson setting and show that they exclude several regularizers covered by the primal analysis of section \ref{sec:convergence-analysis-convex}.

\subsection{Application to Poisson inverse problems}
\label{appendix:sec:non-convex:subsec:poisson}

We now specialize the discussion to the case where $f$ is the Poisson \glsxtrshort{nll}, and where $\varphi$ is the weighted negative entropy defined in \eqref{eq:weighted-negative-entropy}.
In this case, the \glsxtrshort{rbpg} updates \eqref{eq:bpg-right} recover the regularized multiplicative updates \eqref{eq:regularized-multiplicative-updates}.

The first difficulty concerns the existence of the Bregman proximal step associated with the reparameterized regularizer $g \circ \nabla \varphi^*$.
Assumption \ref{assumption:dual-supercoercivity} ensures that its solution set is nonempty and compact through the supercoercivity of $\varphi^* + \tau g \circ \nabla \varphi^*$.
For the weighted negative entropy, however,
\begin{equation*}
  \varphi^*(\eta) = \sum_{i=1}^n s_i \exp\left(\frac{\eta_i}{s_i}\right), \quad \nabla \varphi^* (\eta) = \exp (\eta/s)
\end{equation*}
and $\varphi^*$ is not coercive along directions for which one or more components $\eta_i$ tend to $-\infty$.
In primal coordinates, these directions correspond to $x_i \to 0$.
The reparameterized regularizer $g \circ \nabla \varphi^*$ must therefore compensate for this lack of coercivity.
Equivalently, $g$ must penalize the approach to the boundary of the positive orthant strongly enough that $g \circ \nabla \varphi^*$ grows superlinearly as $\eta_i \to -\infty$.
This is a restrictive requirement.
It excludes many standard regularizers used in imaging, such as the $\ell_1$ norm, Tikhonov regularization, and \glsxtrlong{tv}, which do not blow up at the boundary of the positive orthant.
One example that satisfies this assumption is the indicator $g=\iota_C$ of a nonempty compact convex set $C\subset\mathbb{R}^n_{++}$.
Another example is the convex reciprocal barrier $g(x)=\lambda\sum_{i=1}^n x_i^{-1}$ with $\lambda>0$.

The local regularity assumptions of assumption \ref{assumption:subsequence-convergence} are less problematic.
For the weighted negative entropy, $\varphi^*$ is smooth and locally strongly convex on bounded subsets of $\mathbb{R}^n$.
Moreover, for the Poisson \glsxtrshort{nll}, the reparameterized function $f \circ \nabla \varphi^*$ is smooth on $\mathbb{R}^n$, and its gradient is locally Lipschitz under the standing positivity assumptions on $A$ and $y$.

The principal obstruction to full-sequence convergence is verifying the boundedness of the mirror iterates.
A sufficient condition is the coercivity of the reparameterized objective $\Phi \circ \nabla \varphi^*$.
This condition is again restrictive in the Poisson setting because the Poisson data-fidelity does not necessarily prevent components of $x = \nabla \varphi^*(\eta)$ from approaching zero.
Equivalently, it does not necessarily prevent the corresponding components of $\eta$ from tending to $-\infty$.
Thus, unless the regularizer penalizes the boundary of the positive orthant, the sequence $\left\{ \eta^{(k)} \right\}_{k \in \mathbb{N}}$ may fail to be bounded even when the corresponding sequence $\left\{ x^{(k)} \right\}_{k \in \mathbb{N}}$ remains in a bounded subset of $\mathbb{R}^n_+$.
The standard regularizers listed above do not ensure this property for a general Poisson forward operator.

Finally, full-sequence convergence requires the \glsxtrshort{klp} property of the reparameterized objective function $\Phi \circ \nabla \varphi^*$.
The Poisson data-fidelity is real analytic wherever $Ax>0$.
\begin{lemma}[\glsxtrshort{klp} property of the Poisson \glsxtrshort{nll}]
    The Poisson \glsxtrshort{nll} $f$ defined in \eqref{eq:poisson-nll} is real analytic and satisfies the \glsxtrshort{klp} property on $\mathbb{R}^n_{++}$.
\end{lemma}
\begin{proof}
  Recall that
  \begin{equation*}
    f(x) = \sum_{j=1}^m \left[(Ax)_j - y_j \log (A x)_j\right].
  \end{equation*}
  For each $j$, the univariate function $h_j(z)=z-y_j\log z$ is real analytic on $(0,+\infty)$.
  Since $x\mapsto (Ax)_j$ is linear, each composition $h_j((Ax)_j)$ is real analytic wherever $(Ax)_j>0$.
  Their finite sum is therefore real analytic on $\mathbb{R}^n_{++}$.
  Real-analytic functions are locally subanalytic and satisfy the \glsxtrshort{klp} property, see \cite[Theorem 3.3]{bolteLojasiewiczInequalityNonsmooth2007}.
\end{proof}
The mirror map associated with the weighted negative entropy is also real analytic.
Whenever the complete objective $\Phi$ is real analytic on the positive orthant, the composition $\Phi \circ \nabla \varphi^*$ is real analytic and therefore satisfies the \glsxtrshort{klp} property.

In summary, the non-convex \glsxtrshort{bpg} framework provides a direct route to full-sequence convergence of \glsxtrshort{rbpg}.
For the Poisson model with the weighted negative entropy, however, the required coercivity assumptions force the regularizer to control the approach to the boundary of the positive orthant.
These assumptions exclude several standard imaging regularizers and are therefore more restrictive than those used in the primal convergence analysis of section \ref{sec:convergence-analysis-convex}.

\clearpage
\section{Deconvolution at a higher noise level}
\label{app:high-noise-deconvolution}

We repeat the comparison of section \ref{subsec:deconvolution} with gain $\gamma=10^{-1}$, corresponding to a measurement \glsxtrshort{snr} of approximately 7.54 dB.
We use the same Gaussian blur and 150 iterations, with $\lambda=0.1$ for both Tikhonov and \glsxtrshort{tv} regularization.
The reconstructions are shown in figure \ref{fig:deconvolution-reconstructions-2} and convergence metrics in figure \ref{fig:deconvolution-reconstructions-metrics-2}.
Mirror descent and \texttt{BPG+L2} use the same Armijo backtracking rule \eqref{eq:bpg-l2-armijo} as in the main comparison, with a unit initial trial step size.
The unregularized methods amplify noise, while \texttt{R-BPG+TV} mitigates the stronger noise artifacts and \texttt{OSL+TV} exhibits oscillatory behavior.

\begin{figure}[h]
  \centering
  \includegraphics[width=0.8\textwidth]{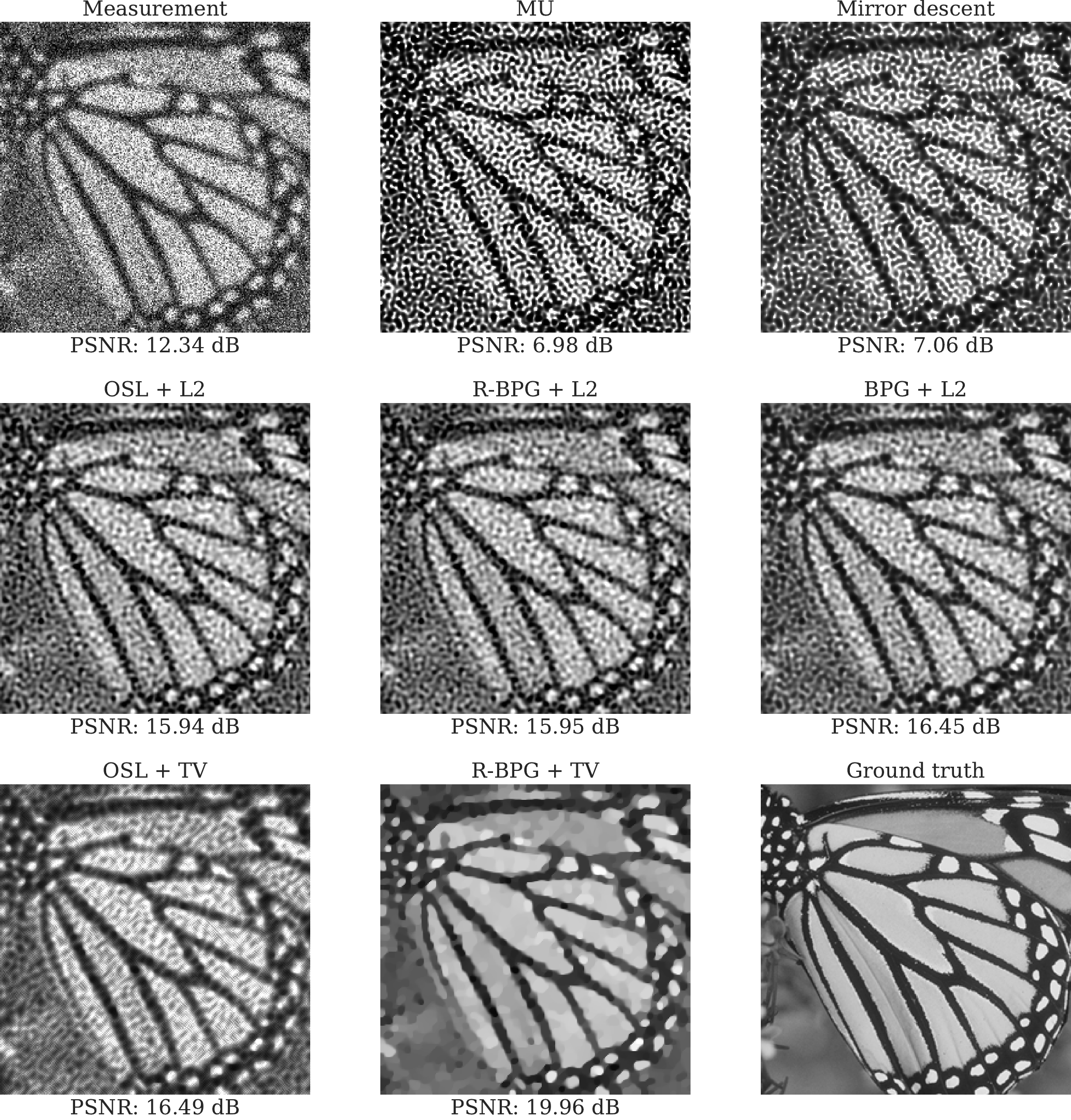}
  \caption{Non-blind deconvolution of the \texttt{butterfly} image from the \texttt{Set3C} dataset, with gain $\gamma = 10^{-1}$. Tikhonov and \glsxtrshort{tv} regularization both use $\lambda=0.1$.}
  \label{fig:deconvolution-reconstructions-2}
\end{figure}

\begin{figure}[h]
  \centering
  \includegraphics[width=\textwidth]{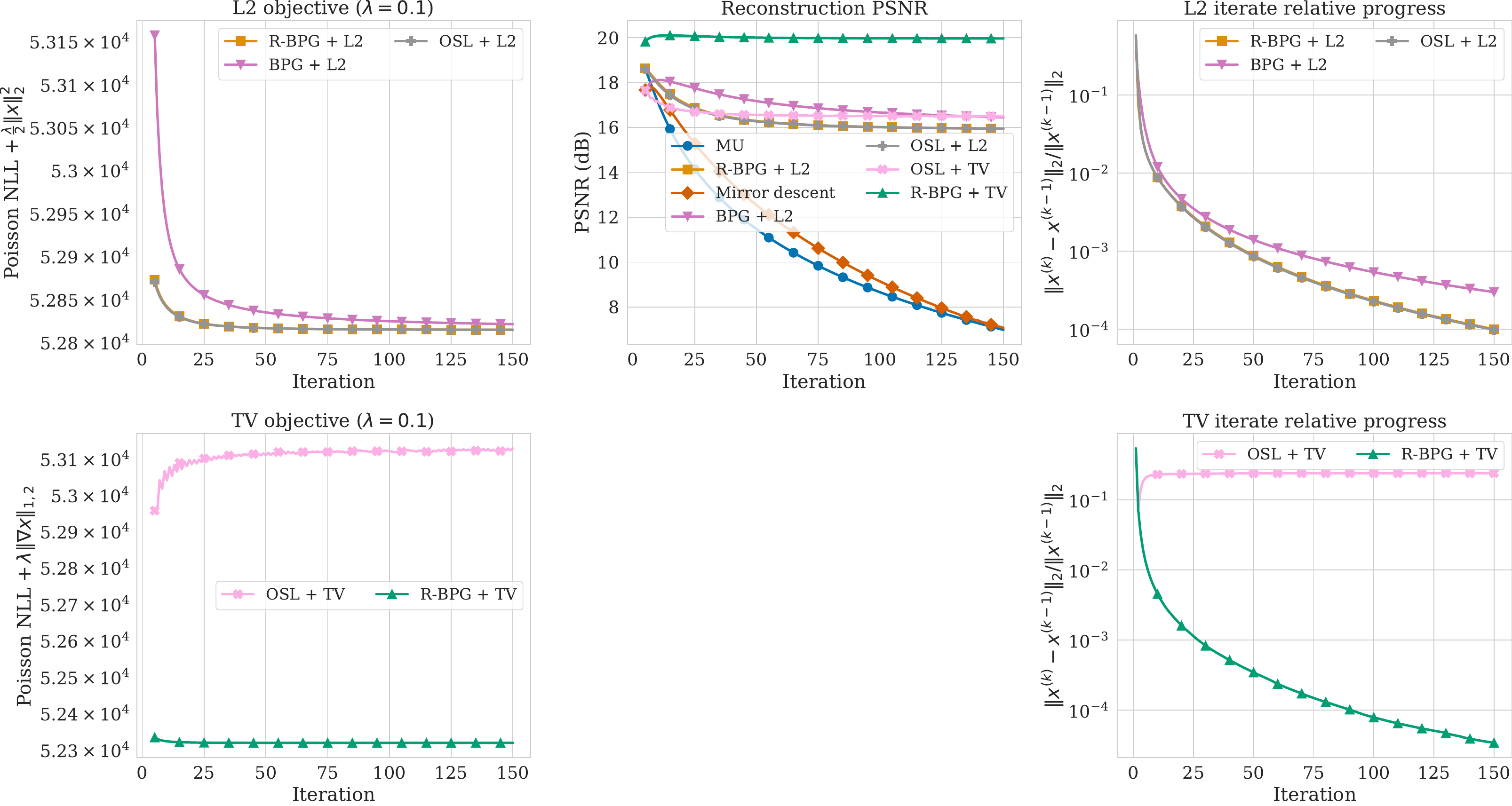}
  \caption{Composite objectives, \glsxtrshort{psnr} in dB, and relative distances between successive iterates for the non-blind deconvolution experiment with gain $\gamma = 10^{-1}$.}
  \label{fig:deconvolution-reconstructions-metrics-2}
\end{figure}

\clearpage
\bibliographystyle{siamplain}
\bibliography{references}

\end{document}